%% file: main.tex
\documentclass[11pt]{article}

\usepackage{times}
\usepackage{mysetup,Shaddin-setup}
\FULLPAGE

\providecommand{\Appendix}{}
\renewcommand{\Appendix}[2][?]{%
	\refstepcounter{section}%
	\vspace{\parskip}%
	{\flushright\Large\bfseries\appendixname\ \thesection: #1}%
	\vspace{\baselineskip}%
}

\renewcommand{\appendix}{%
	\renewcommand{\section}{\secdef\Appendix\Appendix}%
	\renewcommand{\thesection}{\Alph{section}}%
	\setcounter{section}{0}%
}

\newcommand{\UCBcap}{\mathtt{CappedUCB}}        
\newcommand{\Regret}{\mathtt{Regret}}           
\newcommand{\Payoff}{\mathtt{Rev}}        
\newcommand{\RPayoff}{\widehat{\Payoff}}      
\newcommand{\ReservePrice}{p_{\text{r}}}    

\newcommand{\UCB}{\mathtt{UCB1}}
\newcommand{\EXP}{\mathtt{EXP3}}

\newcommand{\ot}{\leftarrow}

\usepackage{amsfonts}
\usepackage{algorithmic}
\usepackage{algorithm}
\usepackage{epsfig}
\usepackage{amscd}
\usepackage{url}
\usepackage{verbatim}

\newcommand{\fkn}{\tfrac{k}{n}}

\title{Dynamic Pricing with Limited Supply%
\footnote{A preliminary version of this paper has appeared in \emph{ACM EC 2012}. An earlier version, titled ``Detail-free, Posted-Price Mechanisms for Limited Supply Online Auctions'',
has appeared in the \emph{Workshop on Bayesian Mechanism Design} at \emph{ACM EC 2011}. The workshop version did not include the results in Section~\ref{sec:LB}.}}

\author{Moshe Babaioff%
\footnote{Microsoft Research Silicon Valley, Mountain View CA, USA.
    Email: {\tt \{moshe,slivkins\} @ microsoft.com}.}
 \and Shaddin Dughmi%
 \footnote{Department of Computer Science,
University of Southern California, Los Angeles CA, USA.
  Email: {\tt shaddin@usc.edu}.}
 \and Robert Kleinberg\footnote{Department of Computer Science, Cornell University, Ithaca NY, USA. Email: {\tt rdk@cs.cornell.edu}.}
 \and Aleksandrs Slivkins\footnotemark[2]
}

\date{First version: July 2011\\
This version: November 2013}
\begin{document}
\maketitle

\begin{abstract}
We consider the problem of designing revenue maximizing online posted-price mechanisms when the seller has limited supply. A seller has $k$ identical items for sale and is facing $n$ potential buyers (``agents'') that are arriving sequentially. Each agent is interested in buying one item. Each agent's value for an item is an independent sample from some fixed (but unknown) distribution with support $[0,1]$. The seller offers a take-it-or-leave-it price to each arriving agent (possibly different for different agents), and aims to maximize his  expected revenue.

We focus on mechanisms that do not use any information about the distribution;
such mechanisms are called \emph{detail-free} (or \emph{prior-independent}). They are desirable because knowing the distribution is unrealistic in many practical scenarios. We study how the revenue of such mechanisms compares to the revenue of the optimal offline mechanism that knows the distribution (``offline benchmark'').

We present a detail-free online posted-price mechanism whose revenue is at most $O((k \log n)^{2/3})$ less than the offline benchmark, for every distribution that is regular. In fact, this guarantee holds without \emph{any} assumptions if the benchmark is relaxed to fixed-price mechanisms. Further, we prove a matching lower bound. The performance guarantee for the same mechanism can be improved to $O(\sqrt{k} \log n)$, with a distribution-dependent constant, if the ratio $\fkn$ is sufficiently small.
We show that, in the worst case over all demand distributions, this is essentially the best rate that can be obtained with a distribution-specific constant.

On a technical level, we exploit the connection to multi-armed bandits (MAB). While dynamic pricing with unlimited supply  can easily be seen as an MAB problem, the intuition behind MAB approaches breaks when applied to the setting with limited supply. Our high-level conceptual contribution is that even the limited supply setting can be fruitfully treated as a bandit problem.
\end{abstract}


\newpage
\input{sec-intro.tex}
\input{sec-prelims.tex}

\input{sec-UCB.tex}
\input{sec-LB.tex}
\input{sec-descending.tex}
\input{sec-conclusions.tex}

\xhdr{Acknowledgements.} We are grateful to Jason Hartline, Qiqi Yan and Assaf Zeevi for their comments.


\bibliographystyle{plain}
\bibliography{bib-abbrv,bib-bandits,bib-rwalks,bib-Shaddin,bib-AGT,bib-slivkins}

\newpage
\appendix
\input{sec-appendix.tex}
\end{document}

%% file: sec-intro.tex
\section{Introduction}
\label{sec:intro}

Consider a promoter that is interested in selling $k$ tickets for a given concert. The seller is interested in maximizing her revenue from selling these tickets, and is offering the tickets on a website such as Ticketmaster. Potential buyers (``agents'') arrive one after another, each with the goal of purchasing a ticket if the price is smaller than the agent's valuation. The seller expects $n$ such agents to arrive. Whenever an agent arrives the seller presents to him a take-it-or-leave-it price, and the agent makes a purchasing decision according to that price. The seller can update the price taking into account the observed history and the number of remaining items and agents.

\OMIT{After each such interaction the seller only learns if the buyer's valuation for a ticket was higher than the posted price, but does not learn the exact valuation of such a buyer.}

We adopt a Bayesian view that the valuations of the buyers are IID samples from a fixed distribution, called \emph{demand distribution}. A standard assumption in a Bayesian setting is that the demand distribution is known to the seller, who can design a specific mechanism tailored to this knowledge. (For example, the Myerson optimal auction for one item sets a reserve price that is a function of the distribution). However, in some settings this assumption is very strong, and should be avoided if possible. For example, when the seller enters a new market, she might not know the demand distribution, and learning it through market research might be costly. Likewise, when the market has experienced a significant recent change, the new demand function might not be easily derived from the old data.

Ideally we would like to design mechanisms that perform well for any demand distribution, and yet do not rely on knowing it.
Such mechanisms are called \emph{detail-free},\footnote{An alternative term used to describe these mechanisms is \emph{prior-independent}.}
in the sense that the specification of the mechanism does not depend on the details of the ``environment'', in the spirit of Wilson's Doctrine~\cite{Wil89}.
Learning about the demand distribution is an integral part of the problem that a detail-free mechanism faces. The performance of such mechanisms is compared to a benchmark that {\em does} depend on the specific demand distribution, as in~\cite{KleinbergL03,HR08,BZ09,DRY10} and many other papers.

In this paper we take this approach and design detail-free, online posted-price mechanisms with revenue that is close to the revenue of the optimal offline mechanism (that can depend on the demand distribution and is not restricted to be posted price). Our main results are for any demand distribution that is regular, or any demand distribution that satisfies the stronger condition of ``monotone hazard rate". Both conditions are mild and standard, and even the stronger one is satisfied by most common distributions, such as the normal, uniform, and exponential distributions.

Posted price mechanisms are commonly used in practice, and are appealing for several reasons. First, an agent only needs to evaluate her offer rather than compute her private value exactly. Human agents tend to find the former task much easier than the latter. Second, agents do not reveal their entire private information to the seller: rather, they only reveal whether their private value is larger than the posted price. Third, posted-price mechanisms are truthful (in dominant strategies) and moreover also group strategy-proof (a notion of collusion resistance when side payments are not allowed).
Further, detail-free posted-price mechanisms are particularly useful in practice as the seller is not required to estimate the demand distribution in advance.
Similar arguments can be found in prior work, e.g.~\cite{ChawlaHMS10}.



\xhdr{Our model.}
We consider the following limited supply auction model, which we term \emph{dynamic pricing with limited supply}. A seller has $k$ items she can sell to a set of $n$ agents (potential buyers), aiming to maximize her expected revenue.
The agents arrive sequentially to the market and the seller interacts with each agent before observing future agents (in an online manner).
We make the simplifying assumption that each agent interacts with the seller only once, and the timing of the interaction cannot be influenced by the agent. (This assumption is also made in other papers that consider our problem for special supply amounts~\cite{KleinbergL03, BBDS11, BZ09}.)
Each agent $i$ ($1\leq i \leq n)$ is interested in buying one item, and has a private value $v_i$ for an item.
The private values are independently drawn from the same {\em demand distribution} $F$. The demand distribution $F$ is {\em unknown} to the seller.
We assume that $F$ has bounded support, and an upper bound on the support is known to the seller;%
\footnote{This assumption enables concentration inequalities such as Chernoff Bounds. It corresponds to the assumption of bounded rewards, which is very common in the literature on multi-armed bandits.}
by normalizing, it is known to the seller that $\mathtt{support}(F)\subset [0,1]$.

\OMIT{
The demand distribution $F$ is {\em unknown} to the seller, but it is known that $F$ has
support in $[0,1]$.\footnote{Assuming that $\mathtt{support}(F)\subset [0,1]$ is w.l.o.g. (by normalizing) as long as the seller knows an upper bound on the support.}
}

Whenever agent $i$ arrives to the market the seller offers him a price $p_i$ for an item. The agent buys the item if and only if $v_i\ge p_i$, and in case she buys the item she pays $p_i$ (so the mechanism is incentive-compatible).
The seller never learns the exact value of $v_i$, she only observes the agent's binary decision to buy the item or not. The seller selects prices $p_i$ using an online algorithm, that we henceforth call \emph{pricing strategy}. We are interested in designing pricing strategies with high revenue compared to a natural benchmark, with minimal assumptions on the demand distribution.

Our main benchmark is the maximal expected revenue of an offline mechanism that is allowed to use the demand distribution; henceforth, we will call it \emph{offline benchmark}.  This is a very strong benchmark, as it has the following advantages over our mechanism: it is allowed to use the demand distribution, it is not constrained to posted prices and is not constrained to run online. It is realized by a well-known Myerson Auction~\cite{Mye81} (which {\em does} rely on knowing the demand distribution).

\OMIT{Our benchmark is the expected revenue of the offline optimal (expected revenue maximizing) mechanism that is allowed to use the demand distribution. Note that the offline mechanism that is optimal is well characterized, it is the Myerson Auction~\cite{Mye81} (and it {\em does} depend on knowledge of the demand distribution). }


\xhdr{High-level discussion.} Absent the supply constraint, our problem fits into the \emph{multi-armed bandit} (MAB) framework~\cite{CesaBL-book}: in each round, an algorithm chooses among a fixed set of alternatives (``arms'') and observes a payoff, and the objective is to maximize the total payoff over a given time horizon.
Our setting corresponds to (prior-free) MAB with \emph{stochastic payoffs}~\cite{Lai-Robbins-85}: in each round, the payoff is an independent sample from some unknown distribution that depends on the chosen ``arm'' (price).
This connection is exploited in~\cite{KleinbergL03,Blum03} for the special case of unlimited supply ($k=n$). The authors use a standard algorithm for MAB with stochastic payoffs, called $\UCB$~\cite{bandits-ucb1}. Specifically, they focus on the prices $\{ i\delta:\, i\in \N\}$, for some parameter $\delta$, and run $\UCB$ with these prices as ``arms''. The analysis relies on the regret bound from~\cite{bandits-ucb1}.

However, neither the analysis nor the intuition behind $\UCB$ and similar MAB algorithms is directly applicable for the setting with limited supply. Informally, the goal of an MAB algorithm would be to converge to a price $p$ that maximizes the expected per-round revenue $R(p) \triangleq p(1-F(p))$. This is, in general, a wrong approach if the supply is limited: indeed, selling at a price that maximizes $R(\cdot)$ may quickly exhaust the inventory, in which case a higher price would be more profitable.

Our high-level conceptual contribution is showing that even the limited supply setting can be fruitfully treated as a bandit problem. The MAB perspective here is that we focus on the trade-off between \emph{exploration} (acquiring new information) and \emph{exploitation} (taking advantage of the information available so far). In particular, we recover an essential feature of $\UCB$ that it does not separate exploration and exploitation, and instead explores arms (prices) according to a schedule that unceasingly adapts to the observed payoffs. This feature results, both for $\UCB$ and for our algorithm, in a much more efficient exploration of suboptimal arms: very suboptimal arms are chosen very rarely even while they are being ``explored".

\OMIT{Optimizing the explore-exploit trade-off in our setting  requires re-inventing some of the MAB techniques, both for algorithms and for the analysis.}

We use an ``index-based" algorithm where each arm is deterministically assigned a numerical score (``index") based on the past history, and in each round an arm with a maximal index is chosen; the index of an arm depends on the past history of this arm (and not on other arms). One key idea is that we define the index of an arm according to the estimated expected total payoff from this arm given the known constraints, rather than according to its estimated expected payoff in a single round. This idea leads to an algorithm that is simple and (we believe) very natural. However, while the algorithm is simple its analysis is not: some new ideas are needed, as the elegant tricks from prior work do not apply (see Section~\ref{sec:UCB} for further discussion).

It is worth noting that a good index-based algorithm did not \emph{have} to exist in our setting. Indeed, many bandit algorithms in the literature are not index-based, e.g. $\EXP$~\cite{bandits-exp3} and ``zooming algorithm"~\cite{LipschitzMAB-stoc08} and their respective variants. The fact that Gittins algorithm~\cite{Gittins-index-79} and $\UCB$~\cite{bandits-ucb1}
achieve (near-)optimal performance with index-based algorithms was widely seen as an impressive contribution.



\xhdr{Contributions.}
In all results below, we consider the dynamic pricing problem with limited supply: $n$ agents and $k\leq n$ items.  We present pricing strategies with expected revenue that is close to the offline benchmark, for large families of natural distributions. All our pricing strategies are deterministic and (trivially) run in polynomial time. Our main result follows.

\begin{theorem}
\label{thm:main}
There exists a detail-free pricing strategy such that for any regular demand distribution its expected revenue is at least the offline benchmark minus
    $O((k \log n)^{2/3})$.
\end{theorem}

We emphasize that Theorem~\ref{thm:main} holds for a pricing strategy that does {\em not} know the demand distribution.
The resulting mechanism is incentive-compatible as it is a posted price mechanism. The specific bound
    $O((k \log n)^{2/3})$
is most informative when
    $k\gg \log n$,
so that the dependence on $n$ is insignificant; the focus here is to optimize the power of $k$. (Note that any non-trivial bound must be below $k$.)

The proof of Theorem~\ref{thm:main} consists of two stages. The first stage (immediate from Yan~\cite{Yan11}) is to observe that for any regular demand distribution the expected revenue of the best fixed-price strategy\footnote{A fixed-price strategy is a pricing strategy that offers the same price to all agents, as long as it has items to sell. The ``best'' fixed-price strategy is one with the maximal expected revenue for a given demand distribution.} is close to the offline benchmark. Henceforth, the expected revenue of the best fixed-price strategy will be called the \emph{fixed-price benchmark}. The second stage, which is our main technical contribution, is to show that our pricing strategy achieves expected revenue that is close to the fixed-price benchmark. Surprisingly, this holds without {\em any} assumptions on the demand distribution.

\begin{theorem}\label{thm:main-fixed}
There exists a detail-free pricing strategy whose expected revenue is at least the fixed-price benchmark minus $O((k \log n)^{2/3})$.
This result holds for every demand distribution. Moreover, this result is the best possible up to a factor of $O(\log n)$.
\end{theorem}

As discussed above, we recover the MAB technique from~\cite{bandits-ucb1} for the unlimited supply setting. The corresponding contribution to the literature on MAB may be of independent interest.

If the demand distribution is regular and moreover the ratio $\fkn$ is sufficiently small then the guarantee in Theorem~\ref{thm:main} can be improved to $O(\sqrt{k} \log n)$, with a distribution-specific constant.

\begin{theorem}\label{thm:MHR-intro}
There exists a detail-free pricing strategy whose expected revenue, for any regular demand distribution $F$, is  at least the offline benchmark minus $O(c_F\, \sqrt{k} \log n)$ whenever $\fkn\leq s_F$, where $c_F$ and $s_F$ are positive constants that depend on $F$. For monotone hazard rate distributions one can take $s_F = \tfrac14$.
\end{theorem}

\OMIT{ 
If the demand distribution is regular and moreover the limited supply constraint is crucial for this distribution -- in the sense that the best fixed price for unlimited supply is \emph{not} the best fixed price for $k$ items -- then the guarantee in Theorem~\ref{thm:main} can be improved to $O(\sqrt{k} \log n)$. Recall that the best fixed price for unlimited supply is
$\argmax_p R(p)$, where $R(p) \triangleq p(1-F(p))$.

\begin{theorem}[Informal]\label{thm:MHR-intro}
Suppose the demand distribution $F$ is regular and moreover $\argmax_p R(p)$ is not the best fixed price for $k$ items. Then there exists a detail-free pricing strategy whose expected revenue is  at least the offline benchmark minus $O(c_F\, \sqrt{k} \log n)$, where $c_F$ is a constant that depends on $F$.
\end{theorem}

In particular, the conditions in Theorem~\ref{thm:MHR-intro} are met if $\tfrac{k}{n}<\tfrac{1}{2e}$ and the demand distribution has the monotone hazard rate (MHR) property, a standard assumption satisfied by many natural distributions.
}

\OMIT{ 
The expected revenue of the pricing strategy in Theorem~\ref{thm:main-fixed} is
within the additive term of $O(c_F\, \sqrt{k} \log n)$ from the offline benchmark
if $k<\tfrac{n}{2e}$ and the demand distribution $F$ satisfies the monotone hazard rate condition, where $c_F$ is a constant that depends on $F$. }

The bound in Theorem~\ref{thm:MHR-intro} is achieved using the pricing strategy from Theorem~\ref{thm:main} with a different parameter. Varying this parameter,  we obtain a family of strategies that improve over the bound in Theorem~\ref{thm:main} in the ``nice" setting of Theorem~\ref{thm:MHR-intro}, and moreover have non-trivial additive guarantees for arbitrary demand distributions. However, we cannot match both theorems with the same parameter.

Note that  the rate-$\sqrt{k}$ dependence on $k$ in Theorem~\ref{thm:MHR-intro} contains a distribution-dependent constant $c_F$ (which can be arbitrarily large, depending on $F$), and thus is not directly comparable to the rate-$k^{2/3}$ dependence in Theorem~\ref{thm:main-fixed}. The distinction (and a significant gap) between bounds with and without distribution-dependent constants is not uncommon in the literature on sequential decision problems, e.g. in~\cite{bandits-ucb1,KleinbergL03,LipschitzMAB-stoc08}.\footnote{For a particularly pronounced example, for the $K$-armed bandit problem with stochastic payoffs the best possible rates for regret with and without a distribution dependent constant are respectively $O(c_F \log n)$ and $O(\sqrt{Kn})$~\cite{bandits-ucb1,bandits-exp3,Bubeck-colt09}.}

In fact, we show that the $c_F\, \sqrt{k}$ dependence on $k$ is essentially the best possible.\footnote{However, the lower bound in Theorem~\ref{thm:LB-intro} does not match the upper bound in Theorem~\ref{thm:MHR-intro} since the latter assumes regularity.} We focus on the fixed-price benchmark (which is a weaker benchmark, so it gives to a stronger lower bound). Following the literature, we define \emph{regret} as the fixed-price benchmark minus the expected revenue of our pricing strategy.

\begin{theorem}\label{thm:LB-intro}
For any $\gamma<\tfrac12$, no detail-free pricing strategy can achieve regret $O(c_F\, k^\gamma)$ for all demand distributions $F$ and arbitrarily large $k,n$, where the constant $c_F$ can depend on $F$.
\end{theorem}

\OMIT{Suppose for some detail-free pricing strategy the expected revenue is within an additive term of $R(k,n,F)$ from the expected revenue of the best fixed-price strategy, for any demand distribution $F$ and arbitrarily large $k,n$. Then it cannot be that $R(k,n,F) = O(c_F\, k^\gamma)$, $\gamma<\tfrac12$,
where $c_F$ is a constant that depends on $F$.
} 

\OMIT{ 
The fact that the upper bound dependence on $k$ is of the rate $\sqrt{k}$ is particularly interesting due to the matching lower bound from prior work~\cite{KleinbergL03,BZ09} for the case $k=\Omega(n)$.~\footnote{The $\Omega(\sqrt{n})$ lower bounds in~\cite{BZ09,KleinbergL03} hold even if the demand distributions are constrained to satisfy some non-degeneracy and smoothness conditions; the conditions in the two papers are incomparable. The result in~\cite{KleinbergL03} only applies to the case $k=n$.} These bounds provide a strong intuition that, informally, one should not hope for a sufficiently general upper bound that is better than $O(\sqrt{k})$. It should be noted that for some distributions $F$ the constant $c_F$ can be very large.
} 

The bounds in Theorem~\ref{thm:main} and Theorem~\ref{thm:main-fixed} are uninformative when $k=O(\log^2 n)$. We next provide another detail-free, online posted-price mechanism that gives meaningful bounds -- not depending on $n$ -- in the case that $k$ is very small (but bigger than some constant).

\begin{theorem}\label{thm:descending-intro}
There exists a detail-free pricing strategy such that for any MHR demand distribution its expected revenue is at least the offline benchmark minus
    $O({k^{3/4}\poly\log(k)})$.
\end{theorem}

\OMIT{ 
Assuming MHR, we show that its expected revenue is within
    $O({k^{3/4}\poly\log(k)})$
of the maximal expected revenue of any offline mechanism.
} 

\section{Related Work}
\label{subsec:related}


\xhdr{Dynamic pricing.}
Dynamic pricing problems and, more generally, revenue management problems, have a rich literature in Operations Research. A proper survey of this literature is beyond our scope; see~\cite{BZ09} for an overview. The main focus is on parameterized demand distributions, with priors on the parameters.

The study of dynamic pricing with \emph{unknown} demand distribution (without priors) has been initiated in~\cite{Blum03,KleinbergL03}. Several special cases of our setting have been studied in~\cite{KleinbergL03,BBDS11,BZ09}, detailed below.

First, Kleinberg and Leighton~\cite{KleinbergL03} consider the unlimited supply case (building on the earlier work~\cite{Blum03}). Among other results, they study IID valuations, i.e. our setting with $k=n$. They provide upper bounds on regret of order $O(n^{2/3})$ and $O(c_F\, \sqrt{n})$.~\footnote{Throughout this section, we omit the $\log$ factors in regret bounds.}
The latter bound is akin to Theorem~\ref{thm:MHR-intro} in that it assumes a version of regularity, and depends on a distribution-specific constant $c_F$. Further, they prove matching lower bounds which, in particular, imply Theorem~\ref{thm:LB-intro} for the special case of unlimited supply.
\footnote{\label{fn:KL03-LB}The construction in~\cite{KleinbergL03} that proves Theorem~\ref{thm:LB-intro}(a) for the unlimited supply case is contained in the proof of a theorem on \emph{adversarial} valuations, but the construction itself only uses IID valuations.}

On the other extreme, Babaioff et al.~\cite{BBDS11} consider the case that the seller has only one item to sell ($k=1$).
They provide a super-constant multiplicative lower bound for unrestricted demand distribution (with respect to the online optimal mechanism), and a constant-factor approximation assuming MHR.
Note that we also use MHR to derive bounds that apply to the case of a very small $k$.

Besbes and Zeevi~\cite{BZ09} consider a continuous-time version which (when specialized to discrete time) is essentially equivalent to our setting with $k = \Omega(n)$.  They prove a number of upper bounds on regret with respect to the fixed-price benchmark, with guarantees that are inferior to ours. The key distinction is that their pricing strategies separate exploration and exploitation. Assuming that the demand distribution $F(\cdot)$ and its inverse $F^{-1}(\cdot)$ are Lipschitz-continuous, they achieve regret $O(n^{3/4})$. They improve it to $O(n^{2/3})$ if furthermore the demand distributions are parameterized, and to $O(\sqrt{n})$
if this is a single-parameter parametrization. Both results rely on knowing the parametrization: the mechanisms continuously update the estimates of the parameter(s) and revise the current price according to these estimates. The upper bounds in~\cite{BZ09} should be contrasted with our $O(k^{2/3})$ upper bound that applies to an arbitrary $k$ and makes no assumptions on the demand distribution, and the $O(c_F\,\sqrt{k})$ improvement for MHR demand distributions.

Also,~\cite{BZ09} contains an $\Omega(\sqrt{n})$ lower bound for their notion of regret. Essentially, this lower bound compares the best pricing strategy for a given demand distribution to the best (distribution-dependent) pricing strategy for a fictitious environment where in every round the mechanism sells a fractional amount of good. In particular, this lower bound does not have any immediate implications on regret with respect to either of the two benchmarks that we use in this paper.

\xhdr{Online mechanisms.}
The study of online mechanisms was initiated by Lavi and Nisan~\cite{LN00}, who unlike us consider the case that each agent is interested in multiple items, and provide a logarithmic multiplicative approximation. Below we survey only the most relevant papers in this line of work, in addition to the special cases of our setting that we have already discussed.

Several papers~\cite{BHW02,Blum03,KleinbergL03,Blum-soda05} consider online mechanisms with unlimited supply and adversarial valuations (as opposed to limited supply and IID valuations in our setting). The mechanism in the initial paper~\cite{BHW02} requires the agents to submit bids and so is not posted-price. The subsequent work~\cite{Blum03,KleinbergL03,Blum-soda05} provides various improvements. In particular, Blum et al.~\cite{Blum03} (among other results) design a simple \emph{posted-price} mechanism which achieves multiplicative approximation $1+\eps$, for any $\eps>0$, with an additive term that depends on $\eps$. \footnote{This result considers valuations in the range $[1,H]$, and the additive term also depends on $H$.}  Blum and Hartline~\cite{Blum-soda05} use a more elaborate posted-price mechanism to improve the additive term. Kleinberg and Leighton~\cite{KleinbergL03} show that the simple mechanism in~\cite{Blum03} achieves regret $O(n^{2/3})$; moreover, they provide a nearly matching lower bound of $\Omega(n^{2/3})$.

Papers~\cite{HKP04,Devanur-ec09} study online mechanisms for limited supply and IID valuations (same as us), but their mechanisms are not posted-price. Hajiaghayi et al.~\cite{HKP04} consider an online auction model where players arrive and depart online, and may misreport the time period during which they participate in the auction. This makes designing strategy-proof mechanisms more challenging, and as a result their mechanisms achieve a constant multiplicative approximation rather than additive regret.
Devanur and Hartline~\cite{Devanur-ec09} study several variants of the limited-supply mechanism design problem: supply is known or unknown, online or offline. Most related to our paper is their mechanism for limited, known, online supply. This mechanism is based on random sampling and achieves constant (multiplicative) approximation, but is not posted-price. Our mechanism is posted-price and achieves low (additive) regret.

\xhdr{Other work.}
Absent the supply constraint, our problem (and a number of related formulations) fit into the \emph{multi-armed bandit} (MAB) framework.\footnote{To avoid a possible confusion, we note that the supply constraint in our setting may appear similar to the budget constraint in line of work on \emph{budgeted MAB} (see~\cite{Bubeck-alt09,Null-soda09} for details and further references). However, the ``budget" in budgeted MAB is essentially the duration of the experimentation phase ($n$), rather than the number of rounds with positive reward ($k$).} MAB has a rich literature in Statistics, Operations Research, Computer Science and Economics. A proper discussion of this literature is beyond the scope of this paper; a reader can refer to~\cite{Bubeck-survey12,Gittins-book11,CesaBL-book} for background. Most relevant to our specific setting is the work on (prior-free) MAB with stochastic payoffs, e.g.~\cite{Lai-Robbins-85,bandits-ucb1}, and MAB with Lipschitz-continuous stochastic payoffs, e.g.~\cite{agrawal-bandits-95,Bobby-nips04,AuerOS/07,LipschitzMAB-stoc08,xbandits-nips08}. The posted-price mechanisms in~\cite{Blum03,KleinbergL03,Blum-soda05} described above are based on a well-known MAB algorithm~\cite{bandits-exp3} for adversarial payoffs.
The connection between online learning and online mechanisms has been explored in a number of other papers, including~\cite{NSV08,DevanurK09,MechMAB-ec09,Transform-ec10}.

Recently,~\cite{ChawlaHMS10,ChakrabortyEGMM10,Yan11} studied the problem of designing an offline, sequential posted-price mechanisms in Bayesian settings, where the distributions of valuations are not necessarily identical, yet are known to the seller. Chawla et al.~\cite{ChawlaHMS10} provide constant multiplicative approximations. Yan~\cite{Yan11} obtains a multiplicative bound that is optimal for large $k$, and Chakraborty et al.~\cite{ChakrabortyEGMM10} obtain a PTAS for all $k$.

Dynamic pricing is superficially similar to \emph{secretary problems}~\cite{Dynkin-secretary63,BabaioffIK07} in that an algorithm is sequentially interacting with agents, each agent's private value is a single number, and it is not known before this agent arrives. However, in secretary problems the private value is revealed when the agent arrives, whereas in dynamic pricing the algorithm is much more constrained in terms of information: the feedback is only whether there is a sale.


%% file: sec-prelims.tex
\section{Preliminaries}
\label{sec:prelims}

\OMIT{A (monopolist) seller has $k$ items to sell. Potential buyers (agents) arrive online (sequentially): agent $t$ arrives in round $t$. We assume that the seller knows that $n$ buyers will arrive. The seller sells his supply  using an {\em online posted-price} mechanism.  Such a mechanism is essentially an online algorithm (\emph{pricing strategy}) which in each round outputs a price and observes whether there was a sale.}

Throughout, we assume that agents' valuations are drawn independently from a distribution $F$ with support in $[0,1]$, called \emph{demand distribution}. We use $p\in [0,1]$ to denote a price. We let $F(p)$ denote the c.d.f, and $S(p)=1-F(p)$ denote the \emph{sales rate} at price $p$: the probability of making a sale at price $p$. Let $R(p) = p\, S(p)$ denote the \emph{revenue function}: the expected single-round revenue at price $p$ given that there is still at least one item left. The demand distribution $F$ is called \emph{regular} if $F(\cdot)$ is twice differentiable and the revenue function $R(\cdot)$ is concave: $R''(\cdot) \leq 0$.
We call $F$ \emph{strictly regular} if furthermore $R''(\cdot) < 0$. Then $R(p)$ is increasing for $p \leq \ReservePrice$ and decreasing for $p \geq \ReservePrice$, where $\ReservePrice$ is the unique maximizer, known as the \emph{Myerson reserve price} (also known as the \emph{monopoly price}).
Moreover, the sales rate $S(\cdot)$ is strictly decreasing, so the inverse $S^{-1}$ is well-defined. We say $F$ is a \emph{Monotone Hazard Rate (MHR)} distribution if $F(\cdot)$ is twice differentiable and the hazard rate $H(p) \triangleq F'(p)/ S(p)$ is non-decreasing. All MHR distributions are regular.

A \emph{fixed-price strategy} with $n$ agents, $k$ items and price $p$, denoted $\A_k^n(p)$, is a pricing strategy that makes a fixed offer price $p$ to every agent so long as fewer than $k$ items have been sold, and stops afterwards (equivalently, from that point always sets the price to $\infty$). Note that for the unlimited supply case $\A^n_n(p)$ sells $n\,S(p)$ items in expectation.

A pricing strategy is called \emph{detail-free} if it does not use the knowledge of the demand distribution. We are interested in designing detail-free pricing strategies with good performance for {\em every} demand distribution in some (large) family of distributions. We compare our mechanisms to two benchmarks that depend on the demand distribution: the maximal expected revenue of an offline mechanism (the \emph{offline benchmark}), and the maximal expected revenue of a fixed price mechanism (the \emph{fixed-price benchmark}). An offline mechanism that maximizes expected revenue was given in the seminal paper of Myerson~\cite{Mye81}; it is not an online posted price mechanism.

Let $\Payoff(\A)$ be the total expected revenue achieved by mechanism $\A$. We define the \emph{regret} of $\A$ with respect to the fixed-price benchmark as follows:
$    \Regret(\A) \triangleq \textstyle{\max_p}\;
        \Payoff[\A_k^n(p)] - \Payoff(\A)$.
Thus, regret is the additive loss in expected revenue compared to the best fixed-price mechanism.
(Note that the regret of $\A$ could, in principle, be a negative number, since the fixed-price benchmark is not generally the Bayesian optimal pricing strategy for distribution $F$.)


\xhdr{Benchmarks Comparison.}
We observe that for regular demand distributions, the fixed-price benchmark is close to the offline benchmark. This result is immediate from Yan~\cite{Yan11}; we provide a self-contained proof in Appendix~\ref{app:benchmarks}.

\begin{lemma}[Yan~\cite{Yan11}]\label{lm:benchmark}
For each regular demand distribution there exists a fixed-price strategy whose expected revenue is at least the offline benchmark minus
    $O(\sqrt{k})$.
\end{lemma}

\OMIT{ 
In fact, \cite{Yan11} obtains a stronger, multiplicative version of the bound in Lemma \ref{lm:benchmark} which we use in Section \ref{sec:descending} (see Appendix~\ref{app:benchmark}).
}

Lemma \ref{lm:benchmark} implies that any pricing strategy with regret
    $O(R)$, $R = \Omega(\sqrt{k})$
with respect to the fixed-price benchmark has the same asymptotic regret $O(R)$ with respect to the offline benchmark, as long as the demand distribution is regular, and in particular if it is MHR. Therefore, the rest of the paper can focus on the fixed-price benchmark. In particular, our main result, Theorem~\ref{thm:main} for regular distributions, follows from Theorem~\ref{thm:main-fixed} that addresses the fixed-price benchmark.


Furthermore, the expected revenue of a fixed-price mechanism has an easy characterization:

\begin{claim}\label{cl:benchmark-nu}
Let $\A$ be the fixed-price mechanism with price $p$. Let
    $\nu(p) = p \min(k, n\,S(p)))$.
Then
\begin{align}\label{eq:cl-benchmark-nu}
\nu(p) - O(p\sqrt{k\log k}) \leq \Payoff(\A) \leq  \nu(p).
\end{align}
It follows that for a strictly regular demand distribution the bound in Lemma \ref{lm:benchmark} is satisfied for the fixed price
    $p^* = \argmax_p \nu(p) =  \max(\ReservePrice,\, S^{-1}(\fkn))$,
where
    $\ReservePrice = \argmax_p p\, S(p)$
is the Myerson reserve price.
\end{claim}

\begin{proof}
Let us focus on the first inequality in~\eqref{eq:cl-benchmark-nu} (the second one is obvious). Let $X_t$ be the indicator variable of sale in round $t$. Denote $X = \sum_{t=1}^n X_t$ and let $\mu = \E[X]$. Then by Chernoff Bounds (Theorem~\ref{thm:chernoff}(a)) with probability at least $1-\tfrac{1}{k}$
it holds that
    $X\geq \mu - O(\sqrt{\mu \log k})$,
in which case
\begin{align*}
\text{\#sales}
    = \min(k,X)
    \geq \min(k, \mu - O(\sqrt{\mu \log k}))
    \geq \min(k,\mu) - O(\sqrt{k \log k}),
\end{align*}
which implies the claim since $\mu = n\,S(p)$.
\end{proof}

\OMIT{ 
We use this fact in Section~\ref{sec:UCB}. Moreover, we can use it to characterize a near-optimal fixed price. This characterization provides intuition for the rest of the paper, and it is used directly in Section \ref{sec:descending}.

\begin{lemma}\label{lm:benchmark-specific}
It follows that the bound in Lemma \ref{lm:benchmark} is satisfied for the fixed price
    $p^* = \argmax_p \nu(p) =  \max(r,\, S^{-1}(\fkn))$,
where
    $r = \argmax_p p\, S(p)$
is the Myerson reserve price.
\end{lemma}
} 

%% file: sec-UCB.tex
\newcommand{\mP}{\mathcal{P}}                   
\newcommand{\mPsel}{\mathcal{P}_{\text{sel}}}   
\newcommand{\bestP}{p^*_{\text{act}}}           
\newcommand{\bestNu}{\nu^*_{\text{act}}}        

\newcommand{\SurvUB}{S^{\mathtt{UB}}_t(p)}

\section{The main technical result: the upper bound in Theorem~\ref{thm:main-fixed}}
\label{sec:UCB}

\OMIT{ 
This section is devoted to the main technical result: Theorem~\ref{thm:main-fixed} for the fixed-price benchmark. This result is very general, as it makes no assumptions on the demand distribution. Let us restate it here for convenience:


\begin{theorem}\label{thm:main-body}
There exists a detail-free, 
online posted-price mechanism whose regret with respect to the fixed-price benchmark is at most $O(k \log n)^{2/3}$.
\end{theorem}

\begin{note}{Remark.}
Theorem~\ref{thm:main-body} provides a non-trivial bound as long as
    $k>\Omega(p^{-3})(\log^2 n)$,
where $p$ is the price such that $S(p) = \fkn$. This is because by Claim~\ref{cl:benchmark-nu} the expected revenue from the fixed-price mechanism with this price is at least
    $kp - O(\sqrt{k\log k})$.
\end{note}
} 

This section is devoted to the main technical result (the upper bound in Theorem~\ref{thm:main-fixed}) which asserts that there exists a detail-free pricing strategy whose regret with respect to the fixed-price benchmark is at most $O((k \log n)^{2/3})$. This result is very general, as it makes no assumptions on the demand distribution.

As discussed in Section~\ref{sec:intro}, we design an algorithm that carefully optimizes the trade-off between exploration and exploitation. We use an \emph{index-based} algorithm in which each arm is assigned a numerical score, called \emph{index}, so that in each round an arm with the highest index is picked. The index of an arm depends only on the past history of this arm. In prior work on index-based bandit algorithms the index of an arm was defined according to estimated expected payoff from this arm in a single round. Instead, we define the index according to estimated expected \emph{total payoff} from this arm given the constraints.

We apply the above idea to $\UCB$. The index in $\UCB$ is, essentially, the best available Upper Confidence Bound (UCB) on the expected single-round payoff from a given arm. Accordingly, we define a new index, so that the index of a given price corresponds to a UCB on the expected total payoff from this price (i.e., from a fixed-price strategy with this price), given the number of agents and the inventory size. Such index takes into account both the average payoff from this arm (``exploitation'') and the number of samples for this arm (``exploration''), as well as the supply constraint. In particular we recover the appealing property of $\UCB$ that it does not separate ``exploration" and ``exploitation", and instead explores arms (prices) according to a schedule that unceasingly adapts to the observed payoffs.

There are several steps to make this approach more precise. First, while it is tempting to use the current values for the number of agents and the inventory size to define the index, we adopt a non-obvious (but more elegant) design choice to use the original values, i.e. the $n$ and the $k$. Second, since the exact expected total payoff for a given price is hard to quantify, we will instead use a natural approximation thereof provided by $\nu(p)$ in Claim~\ref{cl:benchmark-nu}. In other words, our index will be a UCB on $\nu(p)$. Third, in specifying the UCB we will use non-standard estimator from~\cite{LipschitzMAB-stoc08} to better handle prices with very low sales rate.

The main technical hurdle in the analysis is to ``charge'' each suboptimal price for each time that it is chosen, in a way that the total regret is bounded by the sum of these charges and this sum can be usefully bounded from above.  The analysis of $\UCB$ accomplishes this via simple (but very elegant) tricks which, unfortunately, fail in the limited supply setting.

An additional difficulty comes from the probabilistic nature of the analysis. While we adopt a well-known trick -- we define some high-probability events and assume that these events hold deterministically in the rest of the analysis -- choosing an appropriate collection of events is, in our case, non-trivial. Proving that these events indeed hold with high probability relies on some non-standard tail bounds from prior work.

\OMIT{ 
\subsection{High-level discussion}
\label{subsec-overview}

Absent the supply constraint, our problem fits into the \emph{multi-armed bandit} (MAB) framework: in each round, an algorithm chooses among a fixed set of alternatives (``arms'') and observes a payoff, and the objective is to maximize the total payoff over a given time horizon. Our setting corresponds to MAB with \emph{stochastic payoffs}: in each round, the payoff is an independent sample from some unknown distribution that depends on the chosen ``arm'' (price).
This connection is exploited in~\cite{KleinbergL03} for the special case of unlimited supply ($k=n$). The authors use a standard algorithm for MAB with stochastic payoffs, called $\UCB$~\cite{bandits-ucb1}. Specifically, they focus on the prices $\{ i\delta:\, i\in \N\}$, for some parameter $\delta$, and run $\UCB$ with these prices as ``arms''. The analysis relies on the regret bound from~\cite{bandits-ucb1}.

Unfortunately, neither the analysis nor the intuition behind $\UCB$ and similar MAB algorithms is directly applicable to the setting with limited supply. Informally, the goal of an MAB algorithm is to converge to an arm with the highest average payoff. This, of course, is a wrong approach if the supply is limited: if selling at the ``best price'' quickly exhausts the inventory, then a higher price is more profitable.

Our main conceptual challenge is to recover, for the limited supply setting, the appealing feature of $\UCB$ that it does not separate ``exploration'' (acquiring new information) and ``exploitation'' (taking advantage of the infomation available so far). Instead, $\UCB$ explores arms according to a schedule that continuously adapts to the observed payoffs, and thus ensures that  (very) suboptimal arms are chosen (very) rarely even while they are being ``explored". Following the design of $\UCB$, we would like to assign each arm a numerical score, called \emph{index}, so that an arm with the highest index is picked. In $\UCB$, the index is , essentially, the best available Upper Confidence Bound (UCB) on the expected payoff of this arm; the fact that the index depends on both the average payoff and the number of times this arm has been played so far provides the desired exploration-exploitation combination.

The first step is to define an ``index" that reflects exploration, exploitation, and the limited inventory. The key idea here is to interpret the index of a price $p$ as an UCB on the \emph{total payoff} from this price given the inventory, i.e. the total payoff of a fixed-price strategy with price $p$. A natural approximation for the total payoff is given by $\nu(p)$ in Claim~\ref{cl:benchmark-nu}. We adopted a non-obvious (but more elegant) design choice to define the index via $n$ and $k$ rather than, respectively, the number of agents that haven't yet arrived and the number of items that are still unsold. Further, in defining the UCB we used a non-standard estimator from~\cite{LipschitzMAB-stoc08}.

The main technical hurdle in the analysis is to ``charge'' each suboptimal price for each time that it is chosen, in a way that the total regret is bounded by the sum of these charges and this sum can be usefully bounded from above.  The analysis of $\UCB$ accomplishes this via simple (but very elegant) tricks which, unfortunately, fail in the limited supply setting.
} 

\subsection{Our pricing strategy}

Let us define our pricing strategy, called $\UCBcap$. The pricing strategy is initialized with a set $\mP$ of ``active prices''. In each round $t$, some price $p\in \mP$ is chosen. Namely, for each price $p\in \mP$ we define a numerical score, called \emph{index}, and we pick a price with the highest index, breaking ties arbitrarily. Once $k$ items are sold, $\UCBcap$ sets the price to $\infty$ and never sells any additional item.

Recall from Claim~\ref{cl:benchmark-nu} that the expected revenue from the fixed-price strategy $\A_k^n(p)$ is approximated by
    $\nu(p) \triangleq p\cdot \min(k,\, n\, S(p))$.
In each round $t$, we define the \emph{index} $I_t(p)$ as a UCB on $\nu(p)$:
\begin{align*}
I_t(p) \triangleq p\cdot
    \min (k,\, n\, \SurvUB).
\end{align*}
Here $\SurvUB$ is a UCB on the sales rate $S(p)$, as defined below.


For each $p\in \mP$, let $N_t(p)$ be the number of rounds before $t$ in which price $p$ has been chosen, and let $k_t(p)$ be the number of items sold in these rounds. Then
    $\hat{S}_t(p) \triangleq k_t(p)/N_t(p)$
is the current average sales rate. To avoid division by
zero, we define $\hat{S}_t(p)$ to be equal to 1 when $N_t(p)=0$.
We will define
    $\SurvUB = \hat{S}_t(p)+ r_t(p)$,
where $r_t(p)$ is a \emph{confidence radius}: some number such that
\begin{align}\label{eq:confRad-defn}
|S(p) - \hat{S}_t(p)| \leq r_t(p)
    \quad (\forall\, p\in \mP, t\leq n).
\end{align}
holds with high probability, namely with probability at least $1-n^{-2}$.

We need to define a suitable confidence radius $r_t(p)$, which we want to be as small as possible subject to~\eqref{eq:confRad-defn}. Note that $r_t(p)$ must be defined in terms of quantities that are observable at time $t$, such as $N_t(p)$ and $\hat{S}_t(p)$.  A standard confidence radius used in the literature is (essentially)
$r_t(p) = \sqrt{\tfrac{\Theta(\log n)}{N_t(p)+1}}$.

Instead, we use a more elaborate confidence radius from~\cite{LipschitzMAB-stoc08}:
\begin{align}\label{eq:confRad}
r_t(p)
    \triangleq \frac{\alpha}{N_t(p)+1} +
         \sqrt{\frac{\alpha\, \hat{S}_t(p) }{N_t(p)+1}}, \quad
    \text{for some $\alpha= \Theta(\log n)$}.
\end{align}
The confidence radius in~\eqref{eq:confRad} performs as well as the standard one in the worst case:
    $r_t(p) \leq \sqrt{\tfrac{O(\log n)}{N_t(p)+1}}$,
and much better for very small sales rates:
    $r_t(p) \leq \tfrac{O(\log n)}{N_t(p)+1}$;
see Appendix~\ref{app:chernoff} for a self-contained proof.

\OMIT{ 
We use a more elaborate confidence radius from~\cite{LipschitzMAB-stoc08} which performs better for prices with low sales rate. Specifically, a result in~\cite{LipschitzMAB-stoc08} implies that for some  $\alpha = \Theta(\log n)$, with high probability we have
\begin{align}\label{eq:confRad}
|S(p) - \hat{S}_t(p)|
    \leq r_t(p)
    \triangleq \tfrac{\alpha}{N_t(p)+1} +
         \sqrt{\tfrac{\alpha\, \hat{S}_t(p) }{N_t(p)+1}}
     \leq 3\,\left( \tfrac{\alpha}{N_t(p)+1} +
         \sqrt{\tfrac{\alpha\, S_t(p) }{N_t(p)+1}} \right).
\end{align}
(See Lemma~\ref{lm:my-chernoff} for a self-contained proof.) The confidence radius in~\eqref{eq:confRad} performs as well as the standard one in the worst case:
    $r_t(p) \leq \sqrt{\tfrac{O(\log n)}{N_t(p)+1}}$,
but for very small sales rates we obtain a much better bound:
    $r_t(p) \leq \tfrac{O(\log n)}{N_t(p)+1}$.
} 

To recap, we have
\begin{align}\label{eq:index}
I_t(p) \triangleq p\cdot
    \min (k,\, n\, ( \hat{S}_t(p)+ r_t(p))),
    \text{ where $r_t(p)$ is from~\eqref{eq:confRad}.}
\end{align}

Finally, the active prices are given by
\begin{align}\label{eq:mP-def}
    \mP = \{\delta (1+\delta)^i \in [0,1]:\, i\in\N \},
    \text{~where $\delta\in(0,1)$ is a parameter}.
\end{align}
This completes the specification of $\UCBcap$. See Mechanism~\ref{alg:UCBcapped} for the pseudocode.

\floatname{algorithm}{Mechanism}
\begin{algorithm}[h]
\caption{Pricing strategy $\UCBcap$ for $n$ agents and $k$ items}
\label{alg:UCBcapped}
\begin{algorithmic}[1]
\PARAMETER $\delta \in (0,1)$
\STATE  $\mP \ot \{\delta (1+\delta)^i \in [0,1]:\, i\in\N \}$
    \COMMENT{``active prices"}
\STATE While there is at least one item left, in each round $t$ \\
    ~~~~pick any price $p\in \argmax_{p\in\mP} I_t(p)$,
        where $I_t(p)$ is the ``index" given by~\eqref{eq:index}.
\STATE For all remaining agents, set price $p=\infty$.
\end{algorithmic}
\end{algorithm}

\subsection{Analysis of the pricing strategy}
\label{subsec:analysis}

Our goal is to bound from above the \emph{regret} of $\UCBcap$, which is the difference between the optimal expected revenue of a fixed-price strategy and the expected revenue of $\UCBcap$. We prove that $\UCBcap$ achieves regret $O(k \log n)^{2/3}$ for a suitable choice of parameter $\delta$ in~\refeq{eq:mP-def}.

\begin{lemma}\label{lm:UCBcap}
$\UCBcap$ with parameter
    $\delta = k^{-1/3}\, (\log n)^{2/3}$
achieves regret
    $O(k \log n)^{2/3}$.
\end{lemma}

Since the bound in Lemma~\ref{lm:UCBcap} is trivial for $k<\log^2 n$, we will assume that $k\geq \log^2 n$ from now on.

Note that $\UCBcap$ ``exits'' (sets the price to $\infty$) after it sells $k$ items. For a thought experiment, consider a version of this pricing strategy that does not ``exit'' and continues running as if it has unlimited supply of items; let us call this version $\UCBcap'$. Then the realized revenue of $\UCBcap$ is exactly equal to the realized revenue obtained by $\UCBcap'$ from selling the first $k$ items. Thus from here on we focus on analyzing the latter.


We will use the following notation. Let $X_t$ be the indicator variable of the random event that $\UCBcap'$ makes a sale in round $t$. Note that $X_t$ is a 0-1 random variable with expectation $S(p_t)$, where $p_t$ depends on $X_1,\ldots,X_{t-1}$. Let
    $X \triangleq \sum_{t=1}^n\, X_t$
be the total number of sales if the inventory were unlimited. Note that
    $\E[X] = S \triangleq \sum_{t=1}^n\, S(p_t)$.
Going back to our original algorithm, let $\RPayoff$ denote the realized revenue of $\UCBcap$ (revenue that is realized in a given execution). Then
\begin{align}\label{eq:payoff}
&\RPayoff = \textstyle{\sum_{t=1}^N}\; p_t\, X_t,
\qquad    \text{ where }
    N = \max \{ N\leq n:\, \textstyle{\sum_{t=1}^N} X_t \leq k \}.
\end{align}

\xhdr{High-probability events.} We tame the randomness inherent in the sales $X_t$ by setting up three high-probability events, as described below. In the rest of the analysis, we will argue deterministically under the assumption that these three events hold. It suffices because the expected loss in revenue from the low-probability failure events will be negligible. The three events are summarized in the following claim:

\begin{claim}\label{cl:three-events}
With probability at least $1-n^{-2}$ holds, for each round $t$ and each price $p\in \mP$:
\begin{align}
|S(p) - \hat{S}_t(p)|
    &\leq r_t(p)
    \leq 3\,\left( \tfrac{\alpha}{N_t(p)+1} +
         \sqrt{\tfrac{\alpha\, S_t(p) }{N_t(p)+1}} \right),
         \label{eq:confRad-prop} \\
|X-S| &< O(\sqrt{S\,\log n}+ \log n) \label{eq:X-S}, \\
|\textstyle{\sum_{t=1}^n}\; p_t(X_t-S(p_t))|
    & < O(\sqrt{S\,\log n}+ \log n). \label{eq:tail-Pt}
\end{align}
\end{claim}

The probability bounds on the three events in Claim~\ref{cl:three-events} are derived via appropriate concentration inequalities, some of which are non-standard; see Section~\ref{app:chernoff} for further discussion. In the first event, the left inequality asserts that $r_t(p)$ is a confidence radius, and the right inequality gives the performance guarantee for it. The other two events focus on $\UCBcap'$, and bound the deviation of the total number of sales ($X$) and the realized revenue ($\sum_{t=1}^n p_t\, X_t$) from their respective expectations; importantly, these bound are in terms of $\sqrt{S}$ rather than $\sqrt{n}$.

In the rest of the analysis we will assume that the three events in Claim~\ref{cl:three-events} hold deterministically.

\OMIT{ 
Note that for each $t$, $X_t$ is a 0-1 random variable with expectation $S(p_t)$, where $p_t$ depends on $X_1,\ldots,X_{t-1}$. Let
    $S = \sum_{t=1}^n\, S(p_t)$.
Using the appropriate tail bound (Lemma~\ref{lm:Azuma-sharp} with $\alpha_t\equiv 1$) we obtain that
\begin{align}\label{eq:X-S}
    |X-S| < O(\sqrt{S\,\log n}+ \log n)
\end{align}
holds with probability at least $1-n^{-2}$.

Second, taking Lemma~\ref{lm:Azuma-sharp} with $\alpha_t = p_t$ we obtain that
\begin{align}\label{eq:tail-Pt}
    |\textstyle{\sum_{t=1}^n}\; p_t(X_t-S(p_t))| <
        O(\sqrt{S\,\log n}+ \log n)
\end{align}
holds with probability at least $1-n^{-2}$.

Third, let us prove that~\refeq{eq:confRad} is a confidence radius, i.e. that ~\refeq{eq:confRad-defn} holds with high probability. Indeed, for each price $p\in \mP$, let $\{Z_{i,p}\}_{i\leq n}$ be a family of independent 0-1 random variables with expectation $S(p)$. Without loss of generality, let us pretend that the $i$-th time that price $p$ is selected by the pricing strategy, sale happens if and only if $Z_{i,p}=1$. Then by Lemma~\ref{lm:my-chernoff} after the $i$-th play of price $p$ the bound~\refeq{eq:confRad-defn} holds with probability at least $1-n^{-4}$. Taking the Union Bound over all choices of $i$ and all choices of $p$, we obtain that
~\refeq{eq:confRad-defn} holds with probability at least $1-n^{-2}$ as long as $|\mP|\leq n$ (which is the case for us).

From now on, we will assume that~\refeq{eq:confRad-defn},~\refeq{eq:tail-Pt} and~\refeq{eq:X-S} hold.
} 

\xhdr{Single-round analysis.} Let us analyze what happens in a particular round $t$ of the pricing strategy. Let $p_t$ be the price chosen in round $t$. Let
    $\bestP \in \argmax_{p\in \mP} \nu(p)$
be the best active price according to $\nu(\cdot)$,
and let
    $\bestNu \triangleq \nu(\bestP)$.
Let
    $\Delta(p) \triangleq \max(0, \tfrac{1}{n}\,\bestNu - p\, S(p))$
be our notion of ``badness" of price $p$, compared to the optimal approximate revenue $\nu^*$. We will use this notation throughout the analysis, and eventually we will bound regret in terms of
    $\sum_{p\in \mP}\, \Delta(p)\, N(p)$,
where $N(p)$ is the total number of times price $p$ is chosen.

\begin{claim}
For each price $p\in\mP$ it holds that
\begin{align}
N(p) \, \Delta(p)
    \leq O(\log n)\left( 1 + \fkn\,\tfrac{1}{\Delta(p)}
    \right).
 \label{eq:UBonDeltaN}
\end{align}
\end{claim}

\begin{proof}
By definition~\eqref{eq:confRad-defn} of the confidence radius, for each price $p\in\mP$ and each round $t$ we have
\begin{align}
\nu(p) \leq I_t(p) \leq
    p\cdot
        \min\left(k,\, n\, \left( S(p)+ 2\,r_t(p)\right)\right).
\end{align}
Let us use this to connect each choice $p_t$ with $\bestNu$:
\begin{align*}
\begin{cases}
I_t(p_t)
    \geq I_t(\bestP) \geq \nu(\bestP) \triangleq \bestNu \\
I_t(p_t)
    \leq
    p_t\cdot
        \min\left(k,\, n\, \left( S(p_t)+ 2\,r_t(p_t)\right)\right).
\end{cases}
\end{align*}
Combining these two inequalities, we obtain the key inequality:
\begin{align}\label{eq:key-inequality}
\tfrac{1}{n}\,\bestNu
    \leq p_t \cdot
        \min\left(\fkn,\, S(p_t)+ 2\,r_t(p_t)\right).
\end{align}
There are several consequences for $p_t$ and $\Delta(p_t)$:
\begin{align}~\label{eq:consequences}
\left\{
\begin{array}{lll}
p_t           &\geq& \tfrac{1}{k}\,\bestNu \\
\Delta(p_t)   &\leq& 2\,p_t\, r_t(p_t)\\
\Delta(p_t)>0 &\Rightarrow& S(p_t)<\fkn
\end{array}
\right..
\end{align}
The first two lines in~\eqref{eq:consequences} follow immediately from~\eqref{eq:key-inequality}. To obtain the third line, note that $\Delta(p_t)>0$ implies
    $p_t\, k \geq \bestNu > n\, p_t\, S(p_t)$,
which in turn implies $S(p_t)<\fkn$.

Note that we have not yet used the definition~\refeq{eq:confRad} of the confidence radius. For each price $p=p_t$, let $t$ be the last round in which this price has been selected by the pricing strategy.
Note that $N(p)$ (the total number of times price $p$ is chosen) is equal to $N_t(p)+1$.
Then using the second line in~\refeq{eq:consequences} to bound $\Delta(p)$, Eq.~\eqref{eq:confRad-prop} to bound the confidence radius $r_t(p)$, and the third line in~\eqref{eq:consequences} to bound the sales rate, we obtain:
\begin{align*}
\Delta(p)
    \leq O(p)\times \max\left(
        \tfrac{\log n}{N(p)},\;
        \sqrt{\fkn\,\tfrac{\log n}{N(p)}}
    \right).
\end{align*}
Rearranging the terms, we can bound $N(p)$ in terms of $\Delta(p)$ and obtain~\eqref{eq:UBonDeltaN}.
\end{proof}

\xhdr{Analyzing the total revenue.}
A key step is the following claim that allows us to consider
    $\textstyle{\sum_{t=1}^n}\, p_t\, S(p_t)$
instead of the realized revenue $\RPayoff$, effectively ignoring the capacity constraint. This is where we use the high-probability events
~\eqref{eq:X-S} and~\eqref{eq:tail-Pt}. For brevity, let us denote
    $\beta(S) = O(\sqrt{S\log n} + \log n)$.

\begin{claim}\label{cl:payoff-min}
$\RPayoff
    \geq \min(\bestNu,\; \textstyle{\sum_{t=1}^n}\, p_t\, S(p_t))- \beta(k)$.
\end{claim}

\begin{proof}
Recall that
    $p_t\geq \tfrac{1}{k}\bestNu$
by~\refeq{eq:consequences}. It follows that
    $\RPayoff \geq \bestNu$
whenever
     $\textstyle{\sum_{t=1}^n}\, X_t > k$.
Therefore, if $\RPayoff < \bestNu$
then
    $\textstyle{\sum_{t=1}^n}\, X_t \leq k$
and so
    $\RPayoff = \textstyle{\sum_{t=1}^n}\, p_t\, X_t$.
Thus,  by~\eqref{eq:tail-Pt} it holds that
\begin{align*}
\RPayoff
    \geq \min \left(
        \bestNu,\; \textstyle{\sum_{t=1}^n}\, p_t\, X_t
    \right)
    \geq \min \left(
        \bestNu,\; \textstyle{\sum_{t=1}^n}\, p_t\, S(p_t) - \beta(S)
    \right).
\end{align*}
\noindent So the claim holds when $S\leq k$.
On the other hand, if $S>k$ then by~\eqref{eq:X-S} it holds that
\begin{align*}
X &\geq S - \beta(S) \geq k-\beta(k) \\
\RPayoff
    &\geq \min(k,X)\, (\tfrac{1}{k}\,\bestNu)
    \geq \bestNu - \beta(k). \qedhere
\end{align*}
\end{proof}

In light of Claim~\ref{cl:payoff-min}, we can now focus on
    $\textstyle{\sum_{t=1}^n}\, p_t\, S(p_t)$.
\begin{align}
\textstyle{\sum_{t=1}^n}\; p_t\,S(p_t)
    &\geq \textstyle{\sum_{t=1}^n}\; \tfrac{1}{n}\, \bestNu -\Delta(p_t) \nonumber \\
    &= \bestNu - \textstyle{\sum_{t=1}^n}\; \Delta(p_t) \nonumber \\
    &= \bestNu - \textstyle{\sum_{p\in\mP}}\; \Delta(p)\,N(p).
        \label{eq:UCB-rev1}
\end{align}

Fix a parameter $\eps>0$ to be specified later, and denote
\begin{align*}
\begin{cases}
\mPsel      &\triangleq \{ p\in \mP:\, N(p)\geq 1\} \\
\mP_\eps    &\triangleq \{ p\in \mPsel: \Delta(p) \geq \eps\}
\end{cases}
\end{align*}
to be, respectively, be the set of prices that have been selected at least once and the set of prices of badness at least $\eps$ that have been selected at least once. Plugging~\refeq{eq:UBonDeltaN} into~\eqref{eq:UCB-rev1}, we obtain
\begin{align}
\textstyle{\sum_{p\in\mP}}\, \Delta(p)\,N(p)
    &\leq \textstyle{\sum_{p\in\mPsel\setminus\mP_\eps}}\, \Delta(p)\,N(p)
        + \textstyle{\sum_{p\in\mP_\eps}}\, \Delta(p)\,N(p) \nonumber \\
    &\leq \eps n + O(\log n)\,
        \textstyle{\sum_{p\in\mP_\eps}}\;
        \left(
            1+ \fkn \tfrac{1}{\Delta(p)}
        \right) \nonumber \\
    &\leq \eps n
        + O(\log n)\left(
            |\mP_\eps| +
                \fkn\,
                    \textstyle{\sum_{p\in\mP_\eps}}\;\tfrac{1}{\Delta(p)}
        \right). \label{eq:regret-sum}
\end{align}

Combining~\eqref{eq:UCB-rev1}, ~\eqref{eq:regret-sum} and Claim~\ref{cl:payoff-min} yields a claim that summarizes our findings so far.

\begin{claim}\label{cl:for-any-mP}
For any set $\mP$ of active prices and any parameter $\eps>0$ it holds that
\begin{align*}
\bestNu - \E[\RPayoff]
\leq \eps n
        + O(\log n)\left(
            |\mP_\eps| +
                \fkn\,
                \textstyle{\sum_{p\in\mP_\eps}}\;\tfrac{1}{\Delta(p)}
        \right)
        + \beta(k).
\end{align*}
\end{claim}

Interestingly, this claim holds for any set of active prices. The following claim, however, takes advantage of the fact that the active prices are given by~\refeq{eq:mP-def}.

\begin{claim} $\bestNu \geq \nu^* - \delta k$,
where   $\nu^*\triangleq \max_{p} \nu(p)$.
\end{claim}

\begin{proof}
Let $p^*\in \argmax_{p} \nu(p)$ denote the best fixed price with respect to $\nu(\cdot)$, ties broken arbitrarily. If $p^*\leq \delta$ then $\nu^*\leq \delta k$. Else, letting
    $p_0 = \max\{p\in \mP:\, p\leq p^*\}$
we have
    $p_0/p \geq \tfrac{1}{1+\delta} \geq 1-\delta$,
and so
\begin{align*}
\bestNu
    &\geq \nu(p_0)
    \geq \tfrac{p_0}{p^*}\; \nu(p^*) \geq  \nu^* (1-\delta)
    \geq \nu^* - \delta k. \qedhere
\end{align*}
\end{proof}

\noindent It follows that for any
$\eps>0$ and $\delta\in(0,1)$ we have:
\begin{align} \label{eq:regret-raw-1}
\Regret
    \leq  O(\log n)\left(
            |\mP_\eps| +
                \fkn\,
                \textstyle{\sum_{p\in\mP_\eps}}\;\tfrac{1}{\Delta(p)}
        \right)
    +\eps n + \delta k + \beta(k).
\end{align}

\noindent The rest is a standard computation. Plugging in $\Delta(p)\geq \eps $ for each $p\in \mP_\eps$ in~\refeq{eq:regret-raw-1}, we obtain:
\begin{align*}
\Regret
    \leq  O(|\mP_\eps| \log n)
            \left(1+ \tfrac{1}{\eps}\, \fkn \right)
        +\eps n + \delta k + \beta(k).
\end{align*}

\noindent Note that
    $|\mP| \leq \tfrac{1}{\delta}\, \log n$.
To simplify the computation, we will assume that
$\delta\geq \tfrac{1}{n}$ and
    $\eps = \delta\, \fkn$.
Then
\begin{align} \label{eq:regret-raw}
\Regret
    \leq O\left( \delta k + \tfrac{1}{\delta^2} (\log n)^2 + \sqrt{k\log n} \right).
\end{align}
Finally, it remains to pick $\delta$ to minimize the right-hand side of~\refeq{eq:regret-raw}. Let us simply take $\delta$ such that the first two summands are equal:
     $\delta = k^{-1/3}\, (\log n)^{2/3}$.
Then the two summands are equal to
    $O(k \log n)^{2/3}$.
This completes the proof of Lemma~\ref{lm:UCBcap}.

\input{sec-chernoff.tex}

\section{The $O(\sqrt{k}\log n)$ regret bound (Theorem~\ref{thm:MHR-intro})}
\label{sec:root-k}

\OMIT{
We show that the pricing strategy from Section~\ref{sec:UCB} (with a different parameter) satisfies an improved regret bound, $O(\sqrt{k} \log n)$, if the demand distribution is regular and moreover (informally) the best fixed-price strategy for unlimited supply is not the best fixed-price strategy for $k$ items. The formal condition is that $g'(\fkn)>0$, where
$g(s) \triangleq s\, S^{-1}(s)$
is a function from $[S(1), 1]$ to $[0,1]$ that maps a sales rate to the corresponding revenue. The regret bound depends on a distribution-specific constant, namely
    $1+1/g'(\fkn)$.

\begin{theorem}\label{thm:improved-regret}
If the demand distribution is regular and $g'(\fkn)>0$, then $\UCBcap$ with parameter
    $\delta = k^{-1/2}\,\log(n)$
achieves regret
    $O(\sqrt{k}\log n)(1+1/g'(\fkn))$.
\end{theorem}

\begin{note}{Remark.}
The assumption $g'(\fkn)>0$ is satisfied if $\fkn < \tfrac{1}{2e}$ and the demand distribution has the monotone hazard rate (MHR) property. This easily follows from regularity (which implies $g''(\cdot)\leq 0$) and the fact that for MHR distributions any maximizer of $g(\cdot)$ is at least $\tfrac{1}{e}$ (Claim~\ref{lem:mhr_survival}). Moreover,
$g'(\fkn) \geq \Omega( \inf |g''(\cdot)|)$.
\end{note}
} 

We show that the pricing strategy from Section~\ref{sec:UCB} (with a different parameter) satisfies an improved regret bound, $O(\sqrt{k} \log n)$, if the demand distribution is regular and moreover the ratio $\fkn$ is sufficiently small. The regret bound depends on a distribution-specific constant.

\begin{theorem}\label{thm:improved-regret}
For any regular demand distribution $F$ there exist positive constants $s_F$ and $c_F$ such that $\UCBcap$ with parameter
    $\delta = k^{-1/2}\,\log(n)$
achieves regret
    $O(c_F\,\sqrt{k}\log n)$
whenever $\fkn\leq s_F$. For monotone hazard rate distributions we can take $s_F = \tfrac{1}{4}$.
\end{theorem}

\begin{proof}
Let $g(s) \triangleq s\, S^{-1}(s)$ be a function from $[S(1), 1]$ to $[0,1]$ that maps a sales rate to the corresponding revenue. Regularity implies $g''(\cdot)\leq 0$. Since $g'(0)>0$, we can pick a constant $s_F>0$ such that $C\triangleq g'(s_F)>0$. For monotone hazard rate distributions we can take $s_F=\tfrac14$ because for any maximizer $s$ of $g(\cdot)$ it holds that $s\geq \tfrac{1}{e}$ (see Claim~\ref{lem:mhr_survival}). Now, for any $\fkn\leq s_F$ we have that
    $g'(\fkn)\geq C$.
We will use this to obtain a lower bound on $\Delta(p)$; any such lower bound is absent in the analysis in Section~\ref{sec:UCB}. This improvement results in savings in~\refeq{eq:regret-raw-1}, which in turn implies the claimed regret bound.

We will use the notation from Section~\ref{subsec:analysis}, particularly the ``badness" $\Delta(p)$ and the set $\mP_\eps$ of arms of badness $\geq \eps$ that have been selected at least once.
Note that by regularity $g'(s)\geq C$ for any $s\in (0, \fkn)$. Let $p^* = S^{-1}(\fkn)$  and $p\in \mP_\eps$. By the third line in~\refeq{eq:consequences} it holds that $S(p)<\fkn$ and then $p>p^*$.

First, we claim that
    $S(p) < \tfrac{p^*}{p} \fkn$.
Indeed, this is because
    $p\, S(p) = g(S(p)) < g(\fkn) = p^*\,\fkn$.

Second, we bound $\Delta(p)$ from below:
\begin{align*}
\tfrac{1}{n}\, \bestNu &\geq (1-\delta)\, \tfrac{\nu^*}{n} \geq (1-\delta)\,g(\fkn) \\
\Delta(p)
    &\geq (1-\delta)\,g(\fkn) - g(S(p)) \\
    &\geq [g(\fkn) - g(S(p))] - \delta\, g(\fkn) \\
    &\geq C(\fkn- S(p))  - \delta \fkn\, p^* \\
    &\geq C\,\fkn\, (1- \tfrac{p^*}{p}) - \delta \fkn\, p^* \\
    &\geq C\,\fkn\, (1- \tfrac{p^*}{p}(1+\tfrac{\delta}{C})).
\end{align*}
Since $\mP$ is given by~\refeq{eq:mP-def}, it holds that
$\mP_\eps \subset \{ p^* \alpha\, (1+\delta)^i:\, i\in\N\}$
for some $\alpha\geq 1$. Define
\begin{align*}
\mP' \triangleq \{ p\in \mP_\eps:\, p=  p^* \alpha\, (1+\delta)^i
    \text{ with } i \geq \tfrac{2}{C} \}.
\end{align*}
Then for any $p\in \mP'$ it holds that
    $p/p^*= \alpha (1+\delta)^i \geq 1+i\delta $
and therefore
\begin{align*}
\Delta(p)
    \geq C\,\fkn\, (1- \tfrac{1+\delta/C}{1+i\delta})
    \geq \tfrac{C}{2}\,\fkn\, \tfrac{i\delta}{1+i\delta}.
\end{align*}
Therefore, noting that
    $|\mP'|\leq |\mP|\leq O(\tfrac{1}{\delta}\, \log\tfrac{1}{\delta})$,
we have
\begin{align*}
\fkn\,\textstyle{\sum_{p\in \mP'}}\, \tfrac{1}{\Delta(p)}
    &\leq \tfrac{2}{C}\; \textstyle{\sum_{p\in \mP'}}\, (1+\tfrac{1}{i\delta})
    \leq \tfrac{2}{C}\; (|\mP'| + \tfrac{1}{\delta}\, \log |\mP'|)
    \leq O(\tfrac{1}{C} \tfrac{1}{\delta}\, \log\tfrac{1}{\delta})\\
\textstyle{\sum_{p\in \mP_\eps\setminus \mP'}}\, \tfrac{1}{\Delta(p)}
    &\leq \tfrac{1}{\eps}\,|\mP\setminus \mP'|
    \leq \tfrac{1}{\eps}\,(\tfrac{2}{C}+1).
\end{align*}
Plugging this into~\refeq{eq:regret-raw-1} with
    $\eps = \delta\,\fkn$,
we obtain:
\begin{align}
\fkn\,\textstyle{\sum_{p\in \mP_\eps}}\, \tfrac{1}{\Delta(p)}
    &\leq O(\tfrac{1}{\delta}\, \log\tfrac{1}{\delta})(1+\tfrac{1}{C})
    \nonumber \\
\Regret
    &\leq O(\delta k + \tfrac{1}{\delta}(1+\tfrac{1}{C})(\log n)^2 + \sqrt{k\log n}) \label{eq:improved-regret} \\
    &\leq O(c_F\, \sqrt{k} \log n), \text{ where }
        c_F = 1+1/C. \nonumber
\end{align}
The regret bound~\eqref{eq:improved-regret} improves over the corresponding bound~\eqref{eq:regret-raw} in Section~\ref{sec:UCB}. We obtain the final bound by plugging
    $\delta = k^{-1/2}\,\log n$.
\end{proof}

It is desirable to achieve the bounds in Theorem~\ref{thm:main-fixed} and Theorem~\ref{thm:improved-regret} using the same pricing strategy. Unfortunately, the choice of parameter $\delta$ in Theorem~\ref{thm:improved-regret} results in a trivial $O(k)$ regret guarantee for arbitrary demand distributions (as per Equation~\eqref{eq:regret-raw}). However, varying $\delta$ and using Equations~(\ref{eq:regret-raw}) and (\ref{eq:improved-regret}) we obtain a family of pricing strategies that improve over the bound in Theorem~\ref{thm:main-fixed} for the ``nice" setting in Theorem~\ref{thm:improved-regret}, and moreover have non-trivial regret bounds for arbitrary demand distributions.

\begin{theorem}\label{thm:improved-regret-trade-off}
For each $\gamma\in [\tfrac13, \tfrac12]$, consider pricing strategy $\UCBcap$ with parameter
    $\delta = \tilde{O}(k^{-\gamma})$. This pricing strategy
achieves regret
    $\tilde{O}(k^{1-\gamma}) (1+1/g'(\fkn))$
if the demand distribution is regular and $g'(\fkn)>0$, and regret $\tilde{O}(k^{2\gamma})$ for arbitrary demand distributions.
\end{theorem}

%% file: sec-chernoff.tex
\subsection{Concentration inequalities and the proof of Claim~\ref{cl:three-events}}
\label{app:chernoff}

\newcommand{\CB}{{\sc cb}}

We use an elementary concentration inequality known as {\em Chernoff Bounds}, in a formulation from~\cite{MitzUpfal-book05}.


\begin{theorem}[Chernoff Bounds]
\label{thm:chernoff}
Consider $n$ i.i.d. random variables $X_1 \ldots X_n$ with values in $[0,1]$. Let
    $X = \tfrac{1}{n} \sum_{i=1}^n X_i$ be their average, and let $\mu = \E[X]$. Then:
\begin{OneLiners}
\item[(a)] $\Pr[ |X-\mu| > \delta \mu ] < 2\, e^{-\mu n \delta^2/3} $
	for any $\delta\in (0,1)$.
\item[(b)] $\Pr[ X > a ] < 2^{-an} $
	for any $a>6\mu$.
\end{OneLiners}
\end{theorem}

Further, we use a non-standard corollary from~\cite{LipschitzMAB-stoc08}
\footnote{This is Lemma 4.9 in the full (arXiv) version of~\cite{LipschitzMAB-stoc08}.}
which provides us with a sharper (i.e., smaller) confidence radius when $\mu$ is small; we include the proof for the sake of completeness.

\begin{theorem}[\cite{LipschitzMAB-stoc08}]
\label{lm:my-chernoff}
Consider $n$ i.i.d. random variables $X_1 \ldots X_n$ on $[0,1]$. Let $X$ be their average, and let $\mu = \E[X]$. Then for any $\alpha>0$, letting
$r(\alpha,x)
		= \tfrac{\alpha}{n} + \sqrt{\tfrac{\alpha x}{n}}$,
we have:
\begin{align*}
 \Pr\left[\, |X-\mu| < r(\alpha,X) < 3\,r(\alpha,\mu) \, \right] > 1-e^{-\Omega(\alpha)},
\end{align*}
\end{theorem}

\begin{proof}
First, suppose
	$\mu\geq \tfrac{\alpha}{6n} $. Apply Theorem~\ref{thm:chernoff}(a) with
	$\delta = \tfrac12 \sqrt{\tfrac{\alpha}{6\mu n}}$.
Thus with probability at least $1- e^{-\Omega(\alpha)}$ we have
	$|X-\mu| < \delta\mu \leq \mu/2$.
Plugging in the $\delta$,
\begin{align*}
 |X-\mu|
	< \tfrac12  \sqrt{\tfrac{\alpha \mu}{n}}
	\leq \sqrt{\tfrac{\alpha X}{n}}
	\leq r(\alpha, X) < 1.5\, r(\alpha, \mu).
\end{align*}

Now suppose $\mu< \tfrac{\alpha}{6n}$. Then using Theorem~\ref{thm:chernoff}(b) with $a = \tfrac{\alpha}{n}$, we obtain that with probability at least $1- 2^{-\Omega(\alpha)}$ we have
	$X < \tfrac{\alpha}{n}$,
and therefore
    $|X-\mu| < \tfrac{\alpha}{n} < r(\alpha, X)$
and
$$ |X-\mu| < \tfrac{\alpha}{n} < r(\alpha, X) <  (1+\sqrt{2})\, \tfrac{\alpha}{n}
	< 3\, r(\alpha, \mu). \qedhere
$$
\end{proof}

\begin{proof}[{\bf\em Proof of \refeq{eq:confRad-prop} in Claim~\ref{cl:three-events}}]
For each price $p\in \mP$ let $\{Z_{i,p}\}_{i\leq n}$ be a family of independent 0-1 random variables with expectation $S(p)$. Without loss of generality, let us pretend that the $i$-th time that price $p$ is selected by the pricing strategy, sale happens if and only if $Z_{i,p}=1$. Then by Lemma~\ref{lm:my-chernoff} after the $i$-th play of price $p$ the bound~\refeq{eq:confRad-prop} holds with probability at least $1-n^{-4}$. Taking the Union Bound over all choices of $i$ and all choices of $p$, we obtain that
~\refeq{eq:confRad-prop} holds with probability at least $1-n^{-2}$ as long as $|\mP|\leq n$ (which is the case for us).
\end{proof}


\xhdr{Sharper Azuma-Hoeffding inequality.}
We use a concentration inequality on the sum of $n$ random variables $X_t\in\{0,1\}$ such that each variable $X_t$ is a random coin toss with probability $M_t$ that depends on the previous variables $X_1,\ldots,X_{t-1}$. We are interested in bounding the deviation $|X-M|$, where $X = \sum_t X_t$ and $M = \sum_t M_t$. The well-known Azuma-Hoeffding inequality states that with high probability we have
    $|X-M| \leq O(\sqrt{n \log n})$.
However, we need a sharper high-probability bound:
    $|X-M| \leq O(\sqrt{M \log n})$.
Moreover, we need an extension of such bound which considers deviation
    $|\sum_{t=1}^n \alpha_t(X_t-M_t)|$,
where each multiplier $\alpha_t\in[0,1]$ is determined by $X_1,\ldots,X_{t-1}$.

We use the following concentration inequality from the literature.

\begin{theorem}[Theorem 3.15 in~\cite{McDiarmid-concentration}]
\label{thm:Azuma-with-variance}
Let $Z_1, \ldots, Z_n$ be random variables which take values in $[-1,1]$. Let
    $Z = \sum_{t=1}^n\; Z_t$, $\mu = \E[Z]$.
Let
    $V = \textstyle{\sum_{t=1}^n}\; \text{Var}(Z_t| Z_1,\,\ldots, Z_{t-1})$.
Then for any $a>0, v>0$ we have
\begin{align*}
\Pr\left[
 (|Z-\mu| \geq a)\wedge (V\leq v)
\right] \leq e^{-\Omega(\tfrac{a^2}{v+a})}.
\end{align*}
\end{theorem}

We use the above bound to bound the deviation for
    $|\sum_{t=1}^n \alpha_t(X_t-M_t)|$.

\begin{theorem}\label{lm:Azuma-sharp}
Let $X_1, \ldots, X_n$ be 0-1 random variables.
For each $t$, let $\alpha_t\in[0,1]$ be the multiplier determined by $X_1,\ldots,X_{t-1}$. 
Let 
    $M = \sum_{t=1}^n\;  M_t$,
where
    $M_t = \E[X_t| X_1,\,\ldots, X_{t-1}]$
for each $t$. Then for any $b\geq 1$ the event
\begin{align*}
|\textstyle{\sum_{t=1}^n}\; \alpha_t(X_t-M_t)|
        \leq b ( \sqrt{M\,\log n}+ \log n) .
\end{align*}
holds with probability at least $1-n^{-\Omega(b)}$.
\end{theorem}

\begin{proof}
Let
    $Z_t = X_t-y_t$,
where $y_t\in[0,1]$ is a function of $X_1,\ldots, X_{t-1}$, and let
    $Z = \sum_{t=1}^n\; Z_t$.

We claim that
\begin{align}\label{eq:pf-Azuma-sharp}
\Pr\left[ |\textstyle{\sum_{t=1}^n}\; \alpha_t(Z_t-\E[Z_t])|
        \leq  b ( \sqrt{M\,\log n}+ \log n)
\right] \geq 1-n^{-\Omega(b)},\quad
    \text{for any $b\geq 1$}.
\end{align}

To prove~\eqref{eq:pf-Azuma-sharp}, let $\F_t = \sigma(X_1,\,\ldots, X_t)$ be the $\sigma$-algebra generated by $X_1,\,\ldots, X_t$, and let
    $M_t = \E[X_t| X_1,\,\ldots, X_{t-1}]$.
Then conditional on $\F_{t-1}$, $Z_t$ is a random variable with expectation $M_t - y_t$ and two possible values, $-\alpha_t\, y_t$ and $\alpha_t\,(1-y_t)$, where $\alpha_t$ and $y_t$ are constants. It follows that
    $\text{Var}(Z_t|\F_{t-1}) = \alpha_t^2 (M_t - M^2_t) \leq M_t$,
and therefore
    $V\triangleq \sum_{t=1}^n\, \text{Var}(Z_t|\F_{t-1}) \leq M$.

Taking Theorem~\ref{thm:Azuma-with-variance} with
    $a=b ( \sqrt{v\,\log n}+ \log n)$,
we have that for any $b\geq 1$ the event
\begin{align*}
(|Z-E[Z]| \geq b ( \sqrt{v\,\log n}+ \log n))
     \wedge (V\leq v).
\end{align*}
holds with probability at most $n^{-\Omega(b)}$. Finally, we take the Union Bound over (say) all integer $v$ between $\log n$ and $n$, noting that $V\leq M$. This completes the proof of~\eqref{eq:pf-Azuma-sharp}.

Finally, to prove the theorem take~\eqref{eq:pf-Azuma-sharp} with $y_t = M_t$ and note that $Z_t = X_t-M_t$ and so $\E[Z_t] = 0$.
\end{proof}

\begin{proof}[{\bf\em Proof of \eqref{eq:X-S} and ~\eqref{eq:tail-Pt} in Claim~\ref{cl:three-events}}]
Recall that for each $t$, $X_t$ is a 0-1 random variable with expectation $S(p_t)$, where $p_t$ depends on $X_1,\ldots,X_{t-1}$. Using Lemma~\ref{lm:Azuma-sharp} with $\alpha_t\equiv 1$ we obtain~\eqref{eq:X-S}. Using Lemma~\ref{lm:Azuma-sharp} with $\alpha_t = p_t$ we obtain~\eqref{eq:tail-Pt}.
\end{proof}

\OMIT{We note in passing that~\eqref{eq:pf-Azuma-sharp} implies other useful tail bounds. For instance, by setting $\alpha_t \equiv 1$ and $y_t \equiv 0$ we obtain a tail bound for $|X-E[X]|$, where $X = \sum_{t=1}^n\, X_t$.}

%% file: sec-LB.tex
\section{Lower Bounds}
\label{sec:LB}

We prove two lower bounds on regret over all demand distributions which match the upper bounds in Theorem~\ref{thm:main-fixed} and Theorem~\ref{thm:MHR-intro}, respectively. (Note that the latter upper bound is specific to regular distributions.) Throughout this section, \emph{regret} is with respect to the fixed-price benchmark.

\begin{theorem}\label{thm:LB}
Consider the dynamic pricing problem with limited supply: with $n$ agents and $k\leq n$ items.
\begin{OneLiners}
\item[(a)] No detail-free pricing strategy can achieve regret $o(k^{2/3})$ for arbitrarily large $k,n$.
\item[(b)] For any $\gamma<\tfrac12$, no detail-free pricing strategy can achieve regret $O(c_F\, k^\gamma)$ for all demand distributions $F$ and arbitrarily large $k,n$, where the constant $c_F$ can depend on $F$.
\end{OneLiners}
\end{theorem}

\OMIT{
\footnote{In~\cite{KleinbergL03}, the construction that proves Theorem~\ref{thm:LB}(a) for the unlimited supply case is contained in the proof of a theorem on adversarially chosen valuations (but the construction itself is on IID valuations).}
}

Our proof is a black-box reduction to the unlimited supply case ($k=n$). The unlimited supply case of Theorem~\ref{thm:LB} is proved in~\cite{KleinbergL03} (see Footnote~\ref{fn:KL03-LB} on page~\pageref{fn:KL03-LB}).

\begin{proof}
Suppose that some pricing strategy $\A$ violates part (a). Then there is a sequence $\{k_i, n_i\}_{i\in \N}$, where $k_i\leq n_i$ and $\{k_i\}_{i\in \N}$  is strictly increasing, such that $\A$ achieves regret $o(k^{2/3})$ for all problem instances with $n_i$ agents and $k_i$ items, for each $i\in \N$. To obtain a contradiction, let us use $\A$ to solve the unlimited supply problem with regret $o(n^{2/3})$. Specifically, we will solve problem instances with $k_i/4$ agents, for each $i$.

Fix $i\in \N$ and let $k=k_i$ and $n=n_i$. Consider a problem instance $\I$ with unlimited supply and $k/4$ agents and sales rate $S(\cdot)$. Let $\I'$ be an artificial problem instance with unlimited supply and $n$ agents, so that the first $k/4$ agents in $\I'$ correspond to $\I$. Form an artificial problem instance $\J$ with $k$ items and $n$ agents as follows: in each round, $\A$ outputs a price, then with probability $k/2n$ this price is offered to the next agent in $\I'$, and with the remaining probability there is no interaction with agents in $\I'$ and no sale. Since the demand distribution for $\J$ is a mixture of the ``no sale" event which happens with probability $1-\tfrac{k}{2n}$ and the original demand distribution for $\I$, the sales rate for $J$ is given by
    $S_\J(p) = \tfrac{k}{2n} S(p) $.

Running $\A$ on problem instance $\J$ induces a pricing strategy $\A'$ on the original problem instance $\I$.\footnote{If $\A$ stops before it iterates through all agents in $\I$, the remaining agents in $\I$ are offered a price of $\infty$.} In the rest of the proof we show that
$\A'$ achieves regret $o(k^{2/3})$ on $\I$.

\newcommand{\mE}{\mathcal{E}}

Let $\Payoff_\J(\A)$ and $\RPayoff_\J(\A)$ be, respectively, the expected revenue and the realized revenue of $\A$ on problem instance $\J$. Let $r = \argmax_p p S(p)$ be the Myerson reserve price, and let $\A_r$ be the fixed-price strategy with price $r$. By our assumption, we have that
    $\Payoff_\J(\A) \geq \Payoff_\J(\A_r) - o(k^{2/3})$.
We need to deduce that
    $\Payoff_\I(\A') \geq \Payoff_\I(\A_r) - o(k^{2/3})$.

Let $N$ be the number of rounds in $\J$ in which $\A$ interacts with the agents in $\I'$. With high probability $\tfrac{k}{4}<N<k$. Let us condition on $N$ and the event $\mE_N \triangleq \{k/4<N<k\}$:
\begin{align*}
\E[\, \RPayoff_\J(\A_r)\,|\, N,\mE_N \,] &= NrS(r) \\
\E[\, \RPayoff_\J(\A) - \RPayoff_\I(\A') \,|\, N,\mE_N \,] &\leq
    (N-\tfrac{k}{4})\, rS(r).
\end{align*}
Since $\E[N] = \tfrac{k}{2}$, it follows that
\begin{align*}
\Payoff_\I(\A')
    &\geq \Payoff_\J(\A) - \tfrac{k}{4}\, r S(r) - o(1) \\
    &\geq \Payoff_\J(\A_r) - \tfrac{k}{4}\, r S(r) - o(k^{2/3}) \\
    &= \tfrac{k}{4}\, r S(r) - o(k^{2/3}) \\
    &= \Payoff_\I(\A_r) - o(k^{2/3}),
\end{align*}
as required. The reduction for part (b) proceeds similarly.
\end{proof} 

%% file: sec-descending.tex
\section{Selling very few items: proof of Theorem~\ref{thm:descending-intro}}
\label{sec:descending}

\newcommand{\freps}{\frac{1}{\eps}}
\newcommand{\tfreps}{\tfrac{1}{\eps}}

In this section we target a case when very few items are available for sale (roughly, $k<O(\log^2 n)$), so that the bound in Theorem~\ref{thm:main} becomes trivial. We provide a different pricing strategy  whose regret does not depend on $n$, under the mild assumption of monotone hazard rate.

We rely on the characterization in Claim~\ref{cl:benchmark-nu}: we look for the price
    $p^* = \max(\ReservePrice,\, S^{-1}(\fkn))$,
where $\ReservePrice = \argmax_p p\, S(p)$ is the Myerson reserve price. The pricing strategy proceeds as follows (see Mechanism~\ref{alg:pp} on page~\pageref{alg:pp}). It considers prices
    $p_\ell = (1-\delta)^\ell$, $\ell\in \N$
sequentially in the descending order. For each $\ell$, it offers the price $p_\ell$ to a fixed number of agents. The loop stops once the pricing strategy detects that, essentially, the ``best" $p_\ell$ has been reached: either $S(p_\ell)$ is close to $\fkn$, or we are near a maximum of $p\, S(p)$. Parameters are chosen so as to minimize regret.

\floatname{algorithm}{Mechanism}

\begin{algorithm}[h]
\caption{Descending prices}
\label{alg:pp}
\begin{algorithmic}[1]
\PARAMETER Approximation parameters $\delta,\epsilon \in [0,1]$
\STATE Let $\alpha=\left(\frac{k}{n}\right)^{1-\delta}$,
    $\gamma= \min(\alpha, 1/e)$.
\STATE $\ell\ot 0,\; \ell_{\max}\ot 0,\; R_{\max}\ot 0$.
\REPEAT
\STATE $\ell \ot \ell+1,\;  p_\ell \ot  (1+\delta)^{-\ell} $
\STATE Offer price $p_\ell$ to
    $m=\ceil{\delta\, \frac{n}{\log_{1+\delta} (1/\eps)}}$ agents.
\STATE Let $S_\ell$ be the fraction of them who accept.
\STATE Let $R_\ell=p_\ell S_\ell$ be the average per agent revenue.
\STATE  If $S_\ell \geq  (1+\delta)^{-1} \gamma$ and $R_{\ell} \geq R_{\text{max}}$, \STATE \hspace{3mm} then $R_{\max}\ot R_\ell,\; \ell_{\max}\ot \ell$
\UNTIL{$p_\ell \leq \epsilon$ or $S_\ell \geq  (1+\delta) \alpha$ or  $R_\ell \leq (1+\delta)^{-2} R_{\text{max}}$}
\STATE Offer price $\tilde{p}=p_{\ell}$ so long as unsold items remain.
 \end{algorithmic}
\end{algorithm}

\begin{theorem}\label{thm:descending}
For some parameters $\epsilon$ and $\delta$,
Mechanism~\ref{alg:pp} achieves regret
    $O\left({k^{3/4}\poly\log(k)} \right) $
with respect to the offline benchmark, for any demand distribution that satisfies the monotone hazard rate condition.
\end{theorem}

The rest of this section is devoted to proving Theorem~\ref{thm:descending} for
parameters $\epsilon=k^{-1/4}$ and $\delta = (\frac{1}{k}\, \log k)^{1/4}$. We will assume that the demand distribution is MHR, without further notice.  We derive Theorem~\ref{thm:descending} from the following multiplicative bound; it appears difficult to prove the additive version directly.


\OMIT{ 
\begin{align*}
(1- 2 (\log_{1+\delta} \frac{1}{\epsilon}) e^{-\delta^2 \gamma m /4}) (1-15\delta)
\end{align*}
} 

\begin{lemma}\label{lm:mhr_ub}
Assume $p^* \geq \eps$. Set
    $\delta = \sqrt[4]{\tfrac{1}{k}\, \log k \log
        \tfreps \log\log \tfreps}$.
Then the expected revenue of Mechanism~\ref{alg:pp}
is at least $1-O(\delta)$ fraction of the offline benchmark.
\end{lemma}

\begin{proof}[Proof of Theorem~\ref{thm:descending}]
If $p^*\leq \eps$ then the expected loss in revenue is at most $\epsilon k$. Else by Lemma~\ref{lm:mhr_ub} the expected loss in revenue is at most $O(\delta k)$, where $\delta$ is from Lemma~\ref{lm:mhr_ub}. In both cases  the additive regret compared to the offline benchmark is at most $\max(\eps k, O(k\delta))$. Finally, pick $\eps=k^{-1/4}$.
\end{proof}

\subsection{Proof of Lemma~\ref{lm:mhr_ub}}

We use a multiplicative bound in which fixed-price strategies for limited supply are compared to those for unlimited supply (which in turn can be compared to the offline benchmark using Claim~\ref{cl:compare_myerson}).

\begin{lemma}\label{lm:auction_smoothness}
  Assume the demand distribution is regular. Let $p'\leq p$ be two prices such that $p\geq S^{-1}(k/n)$. Let $n' \leq n$. Then
$\Payoff(\A^{n'}_k(p')) \geq \tfrac{n'}{n} \tfrac{p'}{p} \left(1-\tfrac{1}{\sqrt{2 \pi k}}\right) \Payoff(\A^n_n(p))$.
\end{lemma}

The proof uses a technique from~\cite{Yan11}, see Appendix~\ref{app:benchmarks}. Also, we take advantage of several properties of MHR distributions, detailed in Appendix~\ref{app:distributions}.

We say the exploration phase is \emph{$\delta$-approximate} if
\begin{align*}
S(p_\ell) \geq \gamma
    \;\Rightarrow\;
\tfrac{1}{1+\delta} \leq S_\ell/S(p_\ell) \leq 1+\delta.
\end{align*}

\begin{claim}\label{lem:probsucc}
The exploration phase is $\delta$-approximate with probability at least
    $1- 2\, (\log_{1+\delta} \freps)\, e^{-\delta^2 \gamma m /4}$.
\end{claim}

\begin{proof}
  This follows directly by applying Chernoff bounds (both the upper and lower tail form) to the event that some $S_\ell$ violates the condition, then applying the union bound over all choices of $\ell$.
\end{proof}

\begin{claim}\label{lem:approx_price}
When the exploration phase is $\delta$-approximate, we have  $(1-7\delta)S^{-1}\left(\frac{k}{n}\right)  \leq \tilde{p} \leq p^*$.
\end{claim}
\begin{proof}
It is easy to see that none of the stopping conditions of the exploration phase can be triggered until the price goes below $p^*$. Therefore $\tilde{p} \leq p^*$. For the other inequality observe that, by Claim~\ref{lem:logconcave} it holds that
    $S^{-1}(\alpha) \geq (1-\delta)\,S^{-1}(\fkn)$.
Therefore it suffices to show that $\tilde{p} \geq (1-6\delta)\, S^{-1}(\alpha)$.

Assume for a contradiction that the stopping conditions are not triggered in some phase $\ell$ such that
    $p_{\ell+1} < (1+\delta)^{-6}\, S^{-1} (\alpha)$.
Therefore, at round $\ell$ we have
\begin{align}\label{eq:pl}
  p_\ell &= (1+\delta) p_{\ell+1} <  (1+\delta)^{-5}\, S^{-1}(\alpha)
\end{align}
Examining the stopping conditions, and using our assumption above,
we deduce that:
\begin{align}
S_\ell &< (1+\delta)\alpha   \label{eq:1} \\
R_{\text{max}}/ R_\ell &< (1+\delta)^{2} ,   \label{eq:2}
\end{align}
Combining \eqref{eq:pl} and \eqref{eq:1}, we get
\begin{align}\label{eq:rl}
R_\ell = p_\ell S_\ell < (1+\delta)^{-4} \alpha S^{-1}(\alpha)
\end{align}
Note that, since we chose round $\ell$ such that $p_\ell \ll S^{-1}(\alpha)$, the pricing strategy already encountered some round $t < \ell$ such that $p_t$ is``close'' to $S^{-1}(\alpha)$ -- in particular
\begin{align}\label{eq:ptbound}
(1+\delta)^{-1} S^{-1}(\alpha) \leq p_t \leq S^{-1}(\alpha)
\end{align}
and therefore also
$S(p_t) \geq \alpha$.
Since we assume the exploration phase is $\delta$-approximate, the estimated sales rate at round $t$ satisfies
$S_t \geq (1+\delta)^{-1}S(p_t) \geq (1+\delta)^{-1} \alpha$.
Combining this
with~\eqref{eq:ptbound}, we get that the estimated revenue $R_t$ at round $t$ satisfies
\begin{align}\label{eq:rtbound}
R_t = p_t S_t \geq  (1+\delta)^{-2} \alpha S^{-1}(\alpha)
\end{align}

The value of $R_{\text{max}}$ in round $\ell$  is at least $R_t$. Combining \eqref{eq:rtbound} with \eqref{eq:rl}, this shows that at round $\ell$ we have $\frac{R_{\text{max}}}{R_\ell} > (1+\delta)^2$, contradicting \eqref{eq:2}.
\end{proof}

\begin{claim}\label{lem:approx_rev}
When the exploration phase is $\delta$-approximate, we have
    $ R(\tilde{p}) \geq  (1-7\delta)   R(p^*)$.
\end{claim}
\begin{proof}
By Claim~\ref{lem:approx_price}, we are done when $p^*=S^{-1}\left(\frac{k}{n}\right)$. Therefore, assume $p^*=\ReservePrice$, the Myerson reserve price. It is easy to see that $R(p_{\ell+1}) \geq \tfrac{1}{1+\delta}\; R(p_\ell)$ for each $\ell$.  Let $t$ be the first integer such that $p_{t} \leq p^*=\ReservePrice$. Note that $(1+\delta)^{-1} p^* \leq p_t \leq p^*$. Claim~\ref{lem:approx_price} says that $\tilde{p} \leq p^*=\ReservePrice$, therefore $\tilde{\ell} \geq t$ and by Claim~\ref{lem:mhr_survival} $S(p_{t}) \geq S(\ReservePrice) \geq 1/e \geq \gamma$. It suffices to show that a stopping condition must be triggered before $R(p_\ell)$ gets too small.

Assume for a contradiction that the stopping condition is not triggered by phase $\ell\geq t$, for some $\ell$ such that
    $R(p_{\ell+1}) < (1-7\delta) R(p^*)$.
Since $R$ decreases slowly as described above, it follows that $t < \ell$. Moreover, since we assumed the exploration phase is $\delta$-approximate,
$S_{t}
    \geq \tfrac{1}{1+\delta}\; S(p_{t})
    \geq \tfrac{1}{1+\delta}\; \gamma.$
 Therefore, during phase $\ell$ we have
    $R_{max} \geq R_{t} = S_t p_t \geq (\tfrac{1}{1+\delta})^{2}\;  R(p^*)$.
 Since no stopping condition is triggered for phase $\ell$, it must be that
$R_\ell
    \geq (\tfrac{1}{1+\delta})^2\; R_{max}
    \geq (\tfrac{1}{1+\delta})^4\; R(p^*)$.
Moreover
$R(p_{\ell+1})
    \geq \tfrac{1}{1+\delta}\; R( p_\ell) \geq (\tfrac{1}{1+\delta})^2\; R_\ell
    \geq (\tfrac{1}{1+\delta})^{6} R(p^*)$,
a contradiction.
\end{proof}



We can now complete the proof of Lemma~\ref{lm:mhr_ub}.

We condition on the exploration phase being $\delta$-approximate. Let $n'$ and $k'$ be the number of players and items left after the exploration phase, respectively. In the exploitation phase, we attain expected revenue $\Payoff(A^{n'}_{k'}(\tilde{p}))$. Moreover, in the exploration phase we attained revenue at least $(k-k') \tilde{p}$, since we only used prices greater than or equal to $\tilde{p}$. Therefore, the total expected revenue of our pricing strategy is at least $\Payoff(\A^{n'}_{k'}(\tilde{p})) + (k'-k) \tilde{p} $. It is easy to see that this is at least $\Payoff(\A^{n'}_k(\tilde{p}))$.

It remains to bound the expected revenue of $\A^{n'}_k(\tilde{p})$.
Observe that $\frac{n'}{n} \geq 1-\delta$.
For brevity, denote
    $\beta \triangleq (1-\frac{1}{\sqrt{2 \pi k}})$.

There are two cases. In the first case,
    $p^* = S^{-1}(\fkn)$.
Lemma~\ref{lm:auction_smoothness} and Claim~\ref{lem:approx_price} imply that
\begin{align*}
\Payoff(\A^{n'}_k(\tilde{p}))
    \geq \beta\, \frac{n'}{n} \frac{\tilde{p}}{p^*} \; \Payoff(\A^n_n(p^*))
    \geq \beta\; (1-8\delta) \; \Payoff(\A^n_n(p^*)).
\end{align*}

The second case is $p^*=\ReservePrice$. By Claim~\ref{lem:approx_rev} and unimodality of $R$, we have that
\begin{align*}
\Payoff(\A^n_n(\max(S^{-1}(\fkn),\tilde{p})))
    \geq \Payoff(\A^n_n(\tilde{p}))
        \geq (1-7\delta) \;  \Payoff(A^n_n(p^*)).
\end{align*}
Moreover, using Lemma~\ref{lm:auction_smoothness}, Claim~\ref{lem:approx_price}, and the equation above we show that
\begin{align*}
\Payoff(\A^{n'}_k(\tilde{p}))
    \geq \beta \; (1-8\delta)\; \Payoff(\A^n_n(\max(S^{-1}(\fkn),\tilde{p})))
    \geq \beta\; (1-15\delta) \; \Payoff(\A^n_n(p^*)).
\end{align*}

By Lemma \ref{cl:compare_myerson}, Mechanism~\ref{alg:pp} achieves, in expectation, at least the following fraction of the expected revenue of the offline benchmark:
\begin{align*}
\beta\;(1-O(\delta))\;
    \left( 1- 2 \, \log_{1+\delta}(\tfreps) \exp(-\tfrac{1}{4}\,\delta^2 \gamma m) \right).
\end{align*}
Now, plug $\delta$ into Lemma~\ref{lm:mhr_ub}, and $m$ as defined in the pricing strategy. Note that
    $m= \Theta(\frac{\delta^2 n}{\log 1/\eps})$.
We obtain the final bound replacing $\gamma$ by the lesser quantity $\fkn$, and using the fact that
    $\log_{1+\delta}(x) = \Theta(\frac{1}{\delta}\,\log x)$.


%% file: sec-conclusions.tex
\section{Conclusions and open questions}

We consider dynamic pricing with limited supply and achieve near-optimal performance using an index-based bandit-style algorithm. A key idea in designing this algorithm is that we define the index of an arm (price) according to the estimated expected \emph{total payoff} from this arm given the known constraints.

\OMIT{ 
It is worth noting that a good index-based algorithm did not \emph{have} to exist in our setting. Indeed, many bandit algorithms in the literature are not index-based, e.g. $\EXP$~\cite{bandits-exp3} and ``zooming algorithm"~\cite{LipschitzMAB-stoc08} and their respective variants. The fact that Gittins algorithm~\cite{Gittins-index-79} and $\UCB$~\cite{bandits-ucb1}
achieve (near-)optimal performance with index-based algorithms was widely seen as an impressive contribution.
} 

While in this paper we apply the above key idea to a specific index-based algorithm ($\UCB$), it can be seen as an (informal) general reduction for index-based algorithms for dynamic pricing, from unlimited supply to limited supply. This reduction may help with more general dynamic pricing settings (more on that below), and moreover it might be extended to other bandit-style settings where the ``best arm'' is \emph{not} an arm with the best expected per-round payoff. In particular,~\cite{BanditSurveys-colt13} uses a version of this reduction in the context of adaptive quality control in crowdsourcing.

It is an interesting open question whether a reduction such as above can be made more formal, and which algorithms and which settings it can be applied to. An ambitions conjecture for our setting is that there is a simple black-box reduction from unlimited supply to limited supply that applies to arbitrary ``reasonable'' algorithms. In the full generality this conjecture appears problematic; in particular, some reasonable bandit algorithms such as $\EXP$ are hard-coded to spend a prohibitively large amount of time on exploration.

This paper gives rise to a number of more concrete open questions. The most immediate ones concern extending our upper and lower bounds for, respectively, more general and more specific classes of demand functions. First, it is desirable to extend Theorem~\ref{thm:main} to possibly irregular distributions, i.e. obtain non-trivial regret bounds with respect to the offline benchmark. Second, one wonders whether the optimal $O(c_F\, \sqrt{k})$ regret rate from Theorem~\ref{thm:MHR-intro} can be extended to all regular demand distributions. Third, it is open whether our lower bounds can be strengthened to regular demand distributions.

For arbitrary (possibly irregular) distributions, one can show that, essentially, the appropriate benchmark is a \emph{mixture} of two fixed prices rather than one fixed price. In a recent follow-up work~\cite{BwK-focs13}, the authors design regret-minimizing algorithms that compete with this randomized benchmark. In fact, their results extend to a much more general setting of explore-exploit problems with resource utilization constraints.

Further, it is desirable to extend dynamic pricing with limited supply beyond IID valuations. For example, most results on secretary problems \cite{BabaioffIK07} hold under a weaker assumption: random permutation of adversarially chosen values; however, our results do not immediately extend to this model. More generally, one would like to handle adversarial valuations, or perhaps identify assumptions which make the problem tractable. One natural direction is to bound the number or frequency of changes. An initial result in this direction~\cite{BesbesZeevi-OR11} allows the demand distribution to change at most once, at some point in time that is unknown to the mechanism. Alternatively, one could bound the rate of change. A promising approach here is to apply the reduction from this paper to index-based algorithms for the corresponding bandit setting~\cite{DynamicMAB-colt08,contextualMAB-colt11}.

On a final note, we observe that selling at a given price provides some information about the smaller prices, whereas our algorithms do not directly use this information (and neither does the prior work~\cite{KleinbergL03,BZ09}). Likewise, our algorithms do not update the estimates for a given price using the estimates for other prices using the fact that the sales rate $S(p)$ is non-increasing in the price $p$. It is somewhat surprising that our main algorithm achieves near-optimal regret without taking advantage of this additional information. While such information might help in practice, the extent to which it can possibly help is an open question.

%% file: sec-appendix.tex
\section{Benchmark comparison}
\label{app:benchmarks}

We start with a self-contained proof of a slightly weaker version of Lemma~\ref{lm:benchmark} (which suffices for the purposes of this paper).

\begin{lemma}[Yan~\cite{Yan11}]\label{lm:benchmark-appendix}
For each regular demand distribution there exists a fixed-price strategy whose expected revenue is at least the offline benchmark minus
    $O(\sqrt{k\log k})$.
\end{lemma}

Recall that $A^n_k(p)$ denotes the fixed-price strategy with $k$ items, $n$ agents, and fixed price $p$. Let $M^n_k$ denote the optimal (expected revenue maximizing) offline auction with $n$-players and $k$-items.
As in Claim~\ref{cl:benchmark-nu}, let
    $p^* = \max(\ReservePrice,\, S^{-1}(\fkn))$,
where
    $\ReservePrice = \argmax_p p\, S(p)$
is the Myerson reserve price.

\begin{claim}\label{cl:compare_myerson}
If the demand distribution is regular then
$\Payoff(A^n_n(p^*)) \geq \Payoff(M^n_k)$.
\end{claim}

\begin{proof}[Proof of Claim~\ref{cl:compare_myerson}]
  Let $q_i$ be the probability that $M^n_k$ sells to agent $i$. By symmetry, $q_i=q_j$ for all players $i$ and $j$, so we simply denote this probability by $q$. Let $p= S^{-1}(q)$ be the single price we would need to offer a agent in order to sell to him with probability $q$. Since $R$ is a concave function of the selling probability, Jensen's inequality implies that $R(p)$ is an upper bound on the revenue collected by the Myerson auction from a single agent. Equivalently:
    $n R(p) \geq \Payoff(M^n_k)$.

Now, observe that the expected number of items sold by $M^n_k$ is $nq$. Since $M^n_k$ never sells more than $k$ items, it must be that $q \leq \fkn$. Therefore, $p \geq S^{-1}(\fkn)$. By definition of $p^*$, we deduce that there are two cases: (1)  $p^*=\ReservePrice$, or  (2) $\ReservePrice \leq p^* = S^{-1}(\fkn) \leq p$. In case (1) it is clear that $R(p^*) \geq  R(p)$. In case (2) we get that $R(p^*) \geq R(p)$ since $R(x)$ is decreasing for $x \geq \ReservePrice$. Then
\begin{align*}
\Payoff(A^n_n(p^*)) = n R(p^*) \geq nR(p)
    \geq \Payoff(M^n_k). \quad\qedhere
\end{align*}
\end{proof}

Lemma~\ref{lm:benchmark-appendix} follows from Claim~\ref{cl:compare_myerson} and Claim~\ref{cl:benchmark-nu}
because for $p=p^*$ we have $S(p)\leq \fkn$, and so
\begin{align*}
\nu(p) = p\, \min(k, n\,S(p))
    = np^*\,S(p^*) = \Payoff(A^n_n(p^*))
    \geq \Payoff(M^n_k).
\end{align*}

\xhdr{Multiplicative bounds.}
Further, we derive a multiplicative bound in which fixed-price strategies for limited supply are compared to those for unlimited supply. We use this bound to prove Lemma~\ref{lm:auction_smoothness}.

\begin{claim}\label{cl:correlationgap}
For any regular demand distribution and any $p \geq S^{-1}(\fkn)$ it holds that
\begin{align*}
\Payoff(\A^n_k(p)) \geq \left(1-\tfrac{1}{\sqrt{2 \pi k}}\right) \Payoff(\A^n_n(p)).
\end{align*}
\end{claim}

\begin{proof}
The proof uses a technique from \cite{Yan11}.
As a thought experiment, consider an environment where agent valuations are \emph{correlated} as follows: The joint distribution of agent valuations can be sampled by choosing a set $S'$ of $k$ players uniformly at random, then for each agent in $S'$ sampling from the conditional distribution $F(x) |_{x \geq S^{-1}(k/n)}$, and for each agent not in $S'$ sampling from the conditional distribution $F(x) |_{x < S^{-1}(k/n)}$. Observe that each agent's valuation is distributed according to $F$, yet at any point exactly $k$ players have value exceeding $S^{-1}(k/n)$.

Let $T'$ be the set of players in this correlated environment whose valuation exceeds $p$. The probability of a particular agent being included in $T'$ is $S(p)$, and  $\Ex[|T'|] = n S(p)$. Since $p \geq S^{-1}(k/n)$, it is clear that $T' \sse S'$ and therefore $0 \leq |T'| \leq k$.

Now consider our original environment where each agent's valuation is drawn i.i.d from $F$. Let $T$ be the set of players in this environment whose valuations exceed $p$. The  probability of a agent being included in $T$ is $S(p)$ -- the same as the probability of being included in $T'$. However, each agent is included in $T$ \emph{independently} with probability $S(p)$. As a result, some of the players in $T$ do not win an item -- this happens when $|T| > k$. We can write the revenue of $\A^n_k(p)$ in this i.i.d environment as follows.
\begin{align}\label{eq:4}
\Payoff(\A^n_k(p)) = p\, \Ex[\, \min(|T|,k)\,]
\end{align}

Now, observe that $r(Y) = \min(|Y|,k)$ is the rank function of the $k$-uniform matroid. Moreover, it was shown in \cite{Yan11} that the correlation gap of this function is $\beta\triangleq \left(1-\frac{1}{\sqrt{2 \pi k}} \right)$.
Therefore, since each agent is included in $T$ independently, we know by the definition of the correlation gap and the fact that $T$ and $T'$ have the same marginals that
\begin{align}
  \label{eq:5}
\Ex[r(T)] \geq \beta\; \Ex[r(T')].
\end{align}
Recall that $T'$ is always bounded between $0$ and $k$, therefore $r(T')=|T'|$. Combining 
\eqref{eq:4} and \eqref{eq:5}, we get
\begin{align*}
 \Payoff(A^n_k(p))
    = p\; \Ex[\, \min(|T|,k) \,]
    \geq \beta\; p\, \Ex[ |T'| ]
    = \beta\; p n S(p) )
    = \beta\; \Payoff(\A^n_n(p)) ). \qquad \qedhere
\end{align*}
\end{proof}

\begin{corollary*}[Lemma~\ref{lm:auction_smoothness}]
  Assume the demand distribution is regular. Let $p\leq p'$ be two prices such that $p\geq S^{-1}(k/n)$. Let $n' \leq n$. Then
 $\Payoff(\A^{n'}_k(p')) \geq \tfrac{n'}{n} \tfrac{p'}{p} \left(1-\tfrac{1}{\sqrt{2 \pi k}}\right) \Payoff(\A^n_n(p)).$
\end{corollary*}

\begin{proof}
Observe that $\A^n_k(p')$ sells at least as many items as $\A^n_k(p)$ for every realization of the bids, but at price $p'$ instead of $p$. Therefore $\Payoff(A^n_k(p')) \geq \tfrac{p'}{p} \Payoff(A^n_k(p))$. Combining with Claim \ref{cl:correlationgap} we get that
\[\Payoff(A^n_k(p')) \geq \tfrac{p'}{p} \left(1-\tfrac{1}{\sqrt{2 \pi k}}\right) \Payoff(\A^n_n(p)).\]
Next, a simple (omitted) argument shows that the revenue  $\Payoff(A^n_k(p))$ of a fixed price auction exhibits diminishing marginal returns in the number $n$ of players. Therefore, $\Payoff(A^{n'}_k(p)) \geq \frac{n'}{n} \Payoff(A^n_k(p))$.
\end{proof}

Let us note in passing that Claim~\ref{cl:correlationgap} and Claim~\ref{cl:benchmark-nu} imply a stronger, multiplicative version of Lemma~\ref{lm:benchmark}, which is also immediate from~\cite{Yan11}.

\begin{lemma}[Yan~\cite{Yan11}]\label{lem:Yan11}
Assume that the demand distribution is regular. Then there exists a fixed-price strategy whose expected revenue approximates the offline benchmark up to a factor $1-\tfrac{1}{\sqrt{2\pi k}}$.
\end{lemma}

\section{Monotone Hazard Rate distributions}
\label{app:distributions}

Let us state and prove several properties of Monotone Hazard Rate (MHR) distributions which we use in Section~\ref{sec:root-k} and Section~\ref{sec:descending}. Throughout, for a distribution $F$ we use $F(x)$ to denote the c.d.f, $S(x) = 1-F(x)$ to denote the sales rate, and $f(x)$ to denote the p.d.f.

We begin with a simple known characterization of MHR distributions.
\begin{fact}\label{fact:mhr_logconcave}
A distribution is MHR if and only if $S(\cdot)$ is log-concave (i.e. $\log S(x)$ is a concave function of $x$).
\end{fact}


Next, we bound the sales rate at the Myerson reserve price.
\begin{claim}\label{lem:mhr_survival}
Let $F$ be an MHR distribution with support on $[0,\infty]$, and let $S(x)=1-F(x)$. Let
    $r\in \argmax R(\cdot)$
where $R(x) = x\, S(x)$.
Then
    $ S(r) \geq 1/e$.
\end{claim}
\begin{proof}
  We have
  $R'(r) = S(r) + r S'(r) = 0 $.
  Moreover, by Fact \ref{fact:mhr_logconcave} we deduce that \[ \frac{\log S(r)}{r} \geq \frac{d}{dx} \log(S(x)) |_r = \frac{S'(r)}{S(r)} \] Combining with the previous equality, we have
  $\frac{-1}{r} \leq \frac{\log(S(r))}{r} $
  which is equivalent to $S(r) \geq \frac{1}{e}$.
\end{proof}

We now use log-concavity to bound the sensitivity of the inverse of the sales rate.
\begin{claim}\label{lem:logconcave}
Let $F$ be an MHR distribution with support on $[0,\infty]$, and let $\alpha,\beta \in [0,1]$ with $\beta \geq \alpha$. Then \[S^{-1}(\beta) \geq \frac{\log(\beta)}{\log(\alpha)} S^{-1}(\alpha)\]
\end{claim}
\begin{proof}
By Fact \ref{fact:mhr_logconcave}, $f(x) = \log(S(x))$ is a concave, decreasing function of $x$ such that $f(0)=0$ and $f(x) \to -\infty$ as $x \to \infty$. By Jensen's inequality, for every $a,b \in [0,\infty]$ with $b \leq a$  we have
$ f(b)/ f(a) \leq \tfrac{b}{a} $.
Plugging $a=S^{-1}(\alpha)$ and $b=S^{-1}(\beta)$ into this inequality completes the proof.
\end{proof}
